\documentclass[12pt]{article}
\usepackage[margin=1in]{geometry}
\usepackage{lmodern}
\usepackage[utf8]{inputenc}
\usepackage[english]{babel}
\usepackage{graphicx} % Required for inserting images
\usepackage{mathtools}% Loads amsmath
\usepackage{gensymb}
\usepackage{framed}
\usepackage{amsmath}
\usepackage{amssymb}
\usepackage[dvipsnames]{xcolor}
\usepackage{tcolorbox}
\usepackage{enumitem}
\usepackage[inkscapelatex=false]{svg}
\usepackage{float}
\usepackage{caption}
\usepackage{subcaption}
\usepackage{hyperref}
\usepackage{algorithmic}
\usepackage{algorithm}
\usepackage{amsthm}
\usepackage{booktabs}
\usepackage{csquotes}
\usepackage{siunitx}
\DeclareSIUnit\angstrom{\text{\AA}}
\usepackage{tikz}
\usepackage{wrapfig}
\usepackage{textcomp, gensymb}
\usepackage{setspace}
\usetikzlibrary{positioning, shapes, arrows.meta}
\usepackage[titletoc]{appendix}
\usepackage{afterpage}

\usepackage{tikz}
\usetikzlibrary{shapes.geometric, arrows, positioning, calc}

\newtcolorbox{todolist}[1][]{colback=yellow!10!white, colframe=yellow!50!black, title=To Do, #1}

\newtcolorbox{commentbox}[1][]{colback=blue!10!white, colframe=blue!50!black, title=Comment, #1}

\newtheorem{theorem}{Theorem}[section]
\newtheorem{definition}[theorem]{Definition}

\newtheorem{lemma}[theorem]{Lemma}

\newcommand{\abs}[1]{\left \lvert #1 \right \rvert}

\usepackage{xcolor}
\usepackage{xparse}
\newcommand\rise{{\Delta x}}
\newcommand\twist{{\Delta \theta}}

\title{SHREC: A Spectral Embedding-Based Approach for Ab-Initio Reconstruction of Helical Molecules}

\author{Guy Shapira and Yoel Shkolnisky}
\date{}

\begin{document}

\maketitle

\begin{abstract}
Cryo-electron microscopy (cryo-EM) has emerged as a powerful technique for determining the three-dimensional structures of biological molecules at near-atomic resolution. However, reconstructing helical assemblies presents unique challenges due to their inherent symmetry and the need to determine unknown helical symmetry parameters. Traditional approaches require an accurate initial estimation of these parameters, which is often obtained through trial and error or prior knowledge. These requirements can lead to incorrect reconstructions, limiting the reliability of ab initio helical reconstruction.

In this work, we present SHREC (Spectral Helical REConstruction), an algorithm that directly recovers the projection angles of helical segments from their two-dimensional cryo-EM images, without requiring prior knowledge of helical symmetry parameters. Our approach leverages the insight that projections of helical segments form a one-dimensional manifold, which can be recovered using spectral embedding techniques. Experimental validation on publicly available datasets demonstrates that SHREC achieves high resolution reconstructions while accurately recovering helical parameters, requiring only knowledge of the specimen's axial symmetry group. By eliminating the need for initial symmetry estimates, SHREC offers a more robust and automated pathway for determining helical structures in cryo-EM.
\end{abstract}

\section{Introduction}\label{sec:preliminaries}
Cryo-EM is an imaging technique that enables ascertaining the three-dimensional structures of biological molecules \cite{lyumkis2019challenges,cheng2015primer,singer2020computational}. In cryo-EM, a sample of molecules is swiftly vitrified within a thin layer of ice, thereby preserving their native conformation in random, unknown orientations~\cite{wu2020present,cheng2015primer}. An electron beam is transmitted through this frozen sample, creating two-dimensional (2D) projection images called micrographs. These micrographs have a low signal-to-noise ratio (SNR) due to the weak interaction between electrons and biological specimens, as well as the need to limit radiation exposure to prevent damage \cite{singer2020computational, lyumkis2019challenges}. From these micrographs, individual molecules or molecular subregions -- referred to as particles or particle images -- are identified and extracted as small images. Each image corresponds to a 2D projection of a molecule or a molecular subregion, captured in a distinct but unknown orientation. Due to the noise in these images and their unknown imaging parameters, extensive computational processing -- including particle alignment, classification, and averaging -- is required to enhance the signal-to-noise ratio and prepare the data for three-dimensional (3D) reconstruction~\cite{cheng2015primer,singer2020computational,lyumkis2019challenges}. Subsequent high-resolution 3D structure determination further involves sophisticated modeling and reconstruction algorithms, typically requiring thousands of particle images to reliably infer the molecular structure~\cite{singer2020computational,cheng2015primer}. Modern cryo-EM has achieved near-atomic resolutions, making it a popular technique for studying large molecular complexes in different functional states~\cite{cheng2015primer,wu2020present}. 

% \subsection{Ideal imaging model}\label{sec:imaging-model}
We begin by presenting an idealized mathematical model for images obtained via the cryo-EM process. A molecule is represented by its electrostatic potential \(\psi: \mathbb{R}^3 \to \mathbb{R}\), which we assume to be a function with compact support. In a typical experiment, multiple copies of the molecule are frozen in random orientations, which can be described by rotation matrices \(R_i \in \mathrm{SO}(3)\), where \(\mathrm{SO}(3)\) is the group of 3D rotations.
An ideal 2D projection image (or simply a projection) of the molecule in orientation \(R_i\) is generated by integrating the potential rotated by $R_{i}$ along the $z$~direction. Explicitly, we denote the ideal projection operator that maps \(\psi\) to a 2D image by~\(P_{R_i}\), and the projection image  \(P_{R_i}\psi: \mathbb{R}^2 \to \mathbb{R}\) is then given by 
\begin{equation}\label{eq:ideal-projection-model}
P_{R_i}\psi\left(x,y\right) = \int_{-\infty}^{\infty}{\psi\left(R_i\mathbf{r}\right)\mathrm{d}z},
\end{equation}
where \(\mathbf{r} = (x,y,z)^T\) denotes the coordinates in 3D space and \((x,y)\) are coordinates in the 2D image plane. %Notice that we assume that the integral~\eqref{eq:ideal-projection-model} is well-defined. In Section~\ref{subsec:helical-assemblies}, we restrict \(\psi\) and \(R_i\) in a way that assures the integral is well-defined.

We now state a lemma that will be useful later on. The proof of the lemma can be found in Appendix~\ref{app:additional-proofs}.
\begin{lemma}\label{lemma:reflection-invariance}
Let \(\psi: \mathbb{R}^3 \rightarrow \mathbb{R}\) be a function such that for all \(x \in \mathbb{R}\),
\[
\int_{-\infty}^\infty \int_{-\infty}^\infty |\psi(x, y, z)| \, \mathrm{d}y \, \mathrm{d}z < \infty.
\]
Let \(R \in \mathrm{SO}(3)\) be a rotation matrix, and let \(M: \mathbb{R}^3 \rightarrow \mathbb{R}^3\) denote reflection across the \(xy\)-plane. In coordinates, \(M\) is given by
\begin{equation}\label{eq:mirror-operator}
M(x, y, z) = (x, y, -z),
\end{equation}
or equivalently, by the matrix \(M = \mathrm{diag}(1, 1, -1)\). Define the mirror of \(\psi\), denoted \(\psi^M: \mathbb{R}^3 \rightarrow \mathbb{R}\), as
\begin{equation}\label{eq:mirror-function}
    \psi^M(x, y, z) = \psi(x, y, -z),
\end{equation}
i.e., \(\psi^M = \psi \circ M\). Then, it holds that
\begin{equation}\label{eq:mirror-identity}
    P_{R}\psi = P_{MRM}\psi^M,
\end{equation}
where \(P_R\) is the projection operator defined in~\eqref{eq:ideal-projection-model}.
\end{lemma}

Lemma~\ref{lemma:reflection-invariance} reflects an inherent ambiguity of the orientation and of the underlying structure under the projection operation. We elaborate on the consequences of this ambiguity in Section~\ref{sec:problem-setup}.

%\subsection{Helical assemblies}\label{subsec:helical-assemblies}

As this work focuses on structures we helical symmetry, we start by setting the mathematical notation required to analyze helical structures.

\begin{definition}[Helix]\label{def:helix}
\begin{enumerate}
    \item Let \(P \in \mathbb{R}\setminus \{0\}\), and let \(\psi: \mathbb{R}^3 \rightarrow \mathbb{R}\) such that 
    \begin{equation}\label{eq:integral-condition}
        \sup_{x \in \mathbb{R}} \left( \int_{\mathbb{R}^2} |\psi(x,y,z)|^2 \, \mathrm{d}y \, \mathrm{d}z \right)^{1/2} < \infty.
    \end{equation}
    Then, \(\psi\) is said to be a \textbf{continuous helix} with  \textbf{pitch}~\(P\) and screw axis~$\hat{\mathbf{x}}$, if for all \(\mathbf{r} \in \mathbb{R}^3\) and all \(t \in \mathbb{R}\),
    \begin{equation}\label{eq:helix-definition}
        \psi\left(\mathbf{r}\right) = \psi\left(R_x\bigl(\tfrac{2\pi}{P}t\bigr)\mathbf{r} + t\,\hat{\mathbf{x}}\right),
    \end{equation}
    where \(\hat{\mathbf{x}}\) is a unit vector in the \(x\) direction, and $R_x(\theta)$ is the rotation matrix about the $x$-axis given by
    \begin{equation}\label{eq:rotation-matrix}
    R_x(\theta) = \begin{pmatrix}
    1 & 0 & 0 \\
    0 & \cos\theta & -\sin\theta \\
    0 & \sin\theta & \cos\theta
    \end{pmatrix}.
    \end{equation}
    \item Let \(\rise > 0\) and \(\twist \in [-\pi, \pi)\). A scalar field \(\psi: \mathbb{R}^3 \rightarrow \mathbb{R}\) is a \textbf{discrete helix} with \textbf{rise}~\(\rise\),  \textbf{twist}~\(\twist\), and screw axis~\(\hat{\mathbf{x}}\), if for all \(n \in \mathbb{N}\)
    \begin{equation}
        \psi\left(\mathbf{r}\right) = \psi\left(R_x(n\twist)\mathbf{r} + n\rise\,\hat{\mathbf{x}}\right).
    \end{equation}
\end{enumerate}
\end{definition}
Intuitively, a helix is a function that remains unchanged under a combination of rotation around and translation along a fixed axis. In the continuous case, this symmetry is governed by the pitch~\(P\), which specifies the axial distance required for a full rotation of~$\psi$. In the discrete case, the function remains unchanged under successive screw transformations, characterized by a rise \(\rise\) and a twist \(\twist\). The condition~\eqref{eq:integral-condition} ensures that the integral~\eqref{eq:ideal-projection-model} is well defined on \(\psi\) for all \(R_x(\theta)\) and for all \(t \in \mathbb{R}\).

Throughout this work, we assume for convenience that the screw axis of the helix is the $x$-axis, though any other fixed axis can be used (with the compatible definitions of rotation and translation). In addition, from this point onward, we will focus on the continuous helix unless explicitly stated otherwise. We handle discrete helices in Section~\ref{sec:discrete-helix-angle-recovery}.

\begin{definition}[Helical segment]\label{def:helical-segment}
    Let \(B > 0\). 
    \begin{enumerate}
        \item The \textbf{segment operator} with \textbf{box size}~\(B\) denoted\newline
        \(
            S_B: \mathbb{R} \times L^\infty\left(\mathbb{R}^3\right) \rightarrow L^\infty\left(\left[-\frac{B}{2}, \frac{B}{2}\right]^3\right)
        \) is defined for \(t\in \mathbb{R}\) and \(\psi \in L^\infty(\mathbb{R}^3)\) by
        \begin{equation}
            S_B(t, \psi)(\mathbf{r}) = \psi\left(\mathbf{r} - t\hat{\mathbf{x}}\right),\quad\mathbf{r} \in \left[-\frac{B}{2}, \frac{B}{2}\right]^3.
        \end{equation}
        \item Let \(\psi\) be a helix (either continuous or discrete) as defined in Definition~\ref{def:helix}. A \textbf{helical segment}  of \(\psi\) with box size \(B\) is defined as an element of the image of \(S_B\left(\cdot, \psi\right)\). That is, for all \(t \in \mathbb{R}\), \(S_B\left(t, \psi\right): \left[-\frac{B}{2}, \frac{B}{2}\right]^3 \rightarrow \mathbb{R}\) is a helical segment.
    \end{enumerate}
\end{definition}

The defining property of a helix (Equation~\eqref{eq:helix-definition}) leads to an equivalence between translation along the helical axis and rotation about it. The following lemma formalizes this relationship.

\begin{lemma}[Translation-rotation correspondence]\label{lemma:translation-rotation-correspondence}
Let \(\psi: \mathbb{R}^3\rightarrow\mathbb{R}\) be a helix with pitch~\(P \in \mathbb{R}\setminus\{0\}\) as defined in Definition~\ref{def:helix}. Then, for all \(t \in \mathbb{R}\),
\begin{equation}\label{eq:translation-rotation-correspondence}
    \psi\left(\mathbf{r} - t\hat{\mathbf{x}}\right) = \psi\left(R_x\bigl(\tfrac{2\pi}{P}t\bigr)\mathbf{r}\right).
\end{equation}
\begin{proof}
Let \(\psi: \mathbb{R}^3 \rightarrow \mathbb{R}\) be a continuous helix with pitch \(P \in \mathbb{R}\setminus\{0\}\), and let \(t\in \mathbb{R}\). Define \(\mathbf{r}' = \mathbf{r} - t\hat{\mathbf{x}}\). From~\eqref{eq:helix-definition},
\[
\psi(\mathbf{r}') = \psi\left(R_x\left(\tfrac{2\pi}{P}t\right)\mathbf{r}' + t\hat{\mathbf{x}}\right).
\]
By substituting back \(\mathbf{r}' = \mathbf{r} - t\hat{\mathbf{x}}\), and using the fact that \(R_x(\theta)\hat{\mathbf{x}} = \hat{\mathbf{x}}\) for all \(\theta \in \mathbb{R}\), we get
\begin{align*}
\psi(\mathbf{r} - t\hat{\mathbf{x}}) &= \psi\left(R_x\left(\tfrac{2\pi}{P}t\right)(\mathbf{r} - t\hat{\mathbf{x}}) + t\hat{\mathbf{x}}\right)\\  
&=\psi\left(R_x\left(\tfrac{2\pi}{P}t\right)\mathbf{r} - t R_x\left(\tfrac{2\pi}{P}t\right) \hat{\mathbf{x}} + t\hat{\mathbf{x}}\right)\\
&=\psi\left(R_x\left(\tfrac{2\pi}{P}t\right)\mathbf{r}\right).
\end{align*}
\end{proof}
\end{lemma}
Lemma~\ref{lemma:translation-rotation-correspondence} establishes that translation of a helix along its screw axis is equivalent to rotating it. Building on this, we can define the relative orientation between two segments extracted from the same helix.

\begin{definition}[Segment's angle]\label{def:segment-angle}
    Let \(S_B\left(t_1, \psi\right), S_B\left(t_2, \psi\right): \left[-\frac{B}{2}, \frac{B}{2}\right]^3 \rightarrow \mathbb{R} \) be segments with box size~\(B\) of the helix~\(\psi\). An angle \(\theta\in[0,2\pi)\) is an \textbf{angle between the segments} if it satisfies
    \[
        S_B(t_2,\psi)(\mathbf{r})
        \;=\;
        S_B(t_1,\psi \circ R_x(\theta))\bigl(\mathbf{r}\bigr)   
        \quad\forall\mathbf{r}\in[-\tfrac B2,\tfrac B2]^2.
    \]
\end{definition}

The angle between two segments, as defined in Definition~\ref{def:segment-angle}, is not necessarily unique. However, we can assert its existence and provide a formula for such an angle, as we now show. 

\begin{lemma}[Existence of a segment's angle]\label{lemma:segment-angle-existence}
    Let \(S_B\left(t_1, \psi\right), S_B\left(t_2, \psi\right): \left[-\frac{B}{2}, \frac{B}{2}\right]^3 \rightarrow \mathbb{R} \) be segments with box size~\(B\) of the helix~\(\psi\). 
    There exists at least one  \(\theta \in \left[0, 2\pi\right)\) that satisfies the property from Definition~\eqref{def:segment-angle}. Moreover, such an angle is given by
    \begin{equation}\label{eq:theta-delta-t}
      \theta \;=\;\frac{2\pi}{P}\,(t_2 - t_1).
    \end{equation}
\end{lemma}

\begin{proof}
By Definition~\ref{def:helical-segment},   for $\mathbf{r} \in [-\tfrac B2,\tfrac B2]^3$,
\[
  S_B(t,\psi)(\mathbf r)
  = \psi\bigl(\mathbf r - t\hat{\mathbf x}\bigr).
\]
 Lemma~\ref{lemma:translation-rotation-correspondence} states that~\eqref{eq:helix-definition} holds for any \(t\in\mathbb{R}\) and any \(\mathbf{r} \in \mathbb{R}^3\). Specifically, we may set \(\mathbf{r}' = \mathbf{r} - t_2\hat{\mathbf{x}}\), substitute it into~\eqref{eq:helix-definition}, and obtain
 \begin{equation*}
     \psi\left(\mathbf{r}' - t_2\hat{\mathbf{x}}\right) = \psi\left(R_x\bigl(\tfrac{2\pi}{P}t\bigr)(\mathbf{r}' - t_2\hat{\mathbf{x}}) + t\,\hat{\mathbf{x}}\right),
 \end{equation*}
 for all \(\mathbf{r}'\in\mathbb{R}^3\) and for all \(t \in \mathbb{R}\). Using the fact that \(\hat{\mathbf{x}}\) is invariant under rotation about the \(x\)-axis,
 \begin{equation*}
     \psi\left(\mathbf{r}' - t_2\hat{\mathbf{x}}\right) = \psi\left(R_x\bigl(\tfrac{2\pi}{P}t\bigr)(\mathbf{r}') + (t - t_2)\hat{\mathbf{x}}\right).
 \end{equation*}
 Further substituting~\eqref{eq:theta-delta-t} and \(t = t_2 - t_1\) gives us
 \begin{equation*}
 \psi\left(\mathbf{r}' - t_2\hat{\mathbf{x}}\right) = \psi\left(R_x(\theta)\left(\mathbf{r}'\right) - t_1\hat{\mathbf{x}}\right)
 \end{equation*}
for all \(\mathbf{r}' \in \left[-\tfrac B2,\tfrac B2\right]^3\). 
Restricting both sides to \(\left[-\frac{B}{2}, \frac{B}{2}\right]^3\) and using Definition \ref{def:helical-segment}, we get
\[
  S_B(t_2,\psi)(\mathbf r)
  = S_B(t_1,\psi)\!\bigl(R_x(\theta)\,\mathbf r\bigr).
\]
This proves that \(\theta\) is an angle between the two segments.
\end{proof}

Lemma~\ref{lemma:segment-angle-existence} highlights a fundamental consequence of helical symmetry: any two  segments with the same box size of a continuous helix differ only by a rotation about the screw axis. 

We have established the relationship between three-dimensional helical segments. Since in the context of cryo-EM, we deal with two-dimensional projections, we need to connect these notions to the observed images. We first define the projection of a segment.

\begin{definition}[Segment's projection]\label{def:segment-projection}
  Let \(\psi\) be a continuous helix with pitch~\(P\), let 
  \(S_B(t,\psi):\left[-\tfrac B2,\tfrac B2\right]^3\to\mathbb{R}\) be a helical segment (Definition~\ref{def:helical-segment}), and let \(\theta \in \left[0, 2\pi\right)\). The \textbf{segment's projection} of \(S_B(t,\psi)\), which we denote \(\Pi_B\left(t, \psi\right): \left[-\tfrac B2,\tfrac B2\right]^2\rightarrow \mathbb{R}\), is defined as
  \begin{equation}
      \Pi_B\left(t, \psi\right) = P_{I}S_B(t, \psi),
  \end{equation}
  where \(P_{R}\) is the projection operator defined in~\eqref{eq:ideal-projection-model}, and \(I\) is the identity matrix.
\end{definition}

With the projection of a segment defined, we can now relate the angle between helical segments to their 2D projections.

\begin{definition}[Projection angle]\label{def:projection-angle}
    Let \(S_B\left(t_1, \psi\right), S_B\left(t_2, \psi\right): \left[-\frac{B}{2}, \frac{B}{2}\right]^3 \rightarrow \mathbb{R} \) be segments of the helix~\(\psi\), and let \(\Pi_B\left(t_1, \psi\right), \Pi_B\left(t_2, \psi\right): \left[-\frac{B}{2}, \frac{B}{2}\right]^3 \rightarrow \mathbb{R}\) be the corresponding segments' projections. The \textbf{projection angle} between $\Pi_B\left(t_1, \psi\right)$ and $ \Pi_B\left(t_2, \psi\right)$ is defined as the angle between \(S_B\left(t_1, \psi\right)\) and \(S_B\left(t_2, \psi\right)\), as defined in Definition~\eqref{def:segment-angle}.
\end{definition}

A direct yet important consequence of Definitions \ref{def:segment-projection}, \ref{def:projection-angle}, and of Lemma \ref{lemma:segment-angle-existence} is the following.

\begin{lemma}\label{lemma:viewing-angle}
    Let \(S_B\left(t_0, \psi\right), S_B\left(t, \psi\right): \left[-\frac{B}{2}, \frac{B}{2}\right]^3 \rightarrow \mathbb{R} \) be segments of the helix~\(\psi\) for some \(t_0, t \in \mathbb{R}\). Let \(\theta = \frac{2\pi}{P}(t - t_0)\) be the angle between the segments \(S_B(t_0, \psi)\) and \(S_B(t, \psi)\) (as per Definition~\ref{def:segment-angle} and Lemma~\ref{lemma:segment-angle-existence}). Then,
    \begin{equation}
        \Pi_B(t,\psi) = P_{R_x(\theta)}S_B(t_0,\psi),
    \end{equation}
    where \(R_x\) is the rotation matrix defined in~\eqref{eq:rotation-matrix}.
\end{lemma}

\begin{proof}
Let \(\theta = \tfrac{2\pi}{P}(t - t_0)\). By Lemma~\ref{lemma:segment-angle-existence}, we have
\[
S_B(t,\psi)(\mathbf{r}) = S_B(t_0,\psi)(R_x(\theta)\,\mathbf{r}) \quad \forall\, \mathbf{r} \in [-\tfrac{B}{2}, \tfrac{B}{2}]^3.
\]
Applying the projection operator \(P_I\) to both sides and using~\eqref{eq:ideal-projection-model}, we obtain
\begin{align*}
\Pi_B(t,\psi)(x,y)
&= \int_{-\infty}^{\infty} S_B(t,\psi)(x,y,z)\,\mathrm{d}z
= \int_{-\infty}^{\infty} S_B(t_0,\psi)(R_x(\theta)(x,y,z)^T)\,\mathrm{d}z\\
&=
P_{R_x(\theta)} S_B(t_0,\psi)(x,y).
\end{align*}
\end{proof}

Lemma~\ref{lemma:viewing-angle} formalizes the property that a projection of a helical segment is equivalent to the projection of another segment from a different viewing angle around the screw axis. This concept is illustrated in Figure~\ref{fig:viewing-angle}.

\begin{figure}[H]
    \centering
    \includegraphics[width=0.5\linewidth]{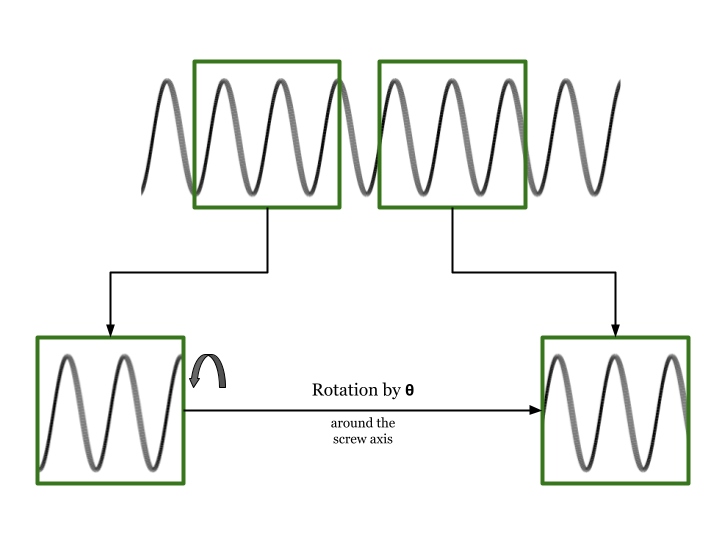}
    \caption{Two helical segments extracted from different positions are related by rotation around the screw axis. Therefore, a 2D projection of one segment corresponds to a projection of the other from a different angle.}
    \label{fig:viewing-angle}
\end{figure}

By selecting an arbitrary reference segment at \(t = t_0\), one may define for each helical segment \(S_B(t,\psi)\) a corresponding \emph{segment angle} -- namely the angle between \(S_B(t,\psi)\) and \(S_B\left(t_0, \psi\right)\). Likewise, we can assign a \emph{projection angle} to its two-dimensional projection \(\Pi_B(t,\psi)\). Since the choice of \(t_0\) is arbitrary, all such angles are determined only up to a global additive constant.

In addition to the fundamental helical symmetry defined by the pitch \(P\) (or rise \(\rise\) and twist \(\twist\) for discrete helices), many helical structures exhibit further symmetries. A common type is axial cyclic symmetry, denoted as \(C_n\) symmetry, where the structure is invariant under discrete rotations around the helical axis.

\begin{definition}[Axial \(C_n\) symmetry]\label{def:cn-symmetry}
Let \(\psi: \mathbb{R}^3 \rightarrow \mathbb{R}\) be a scalar field representing a molecule with the \(x\)-axis as its principal axis (e.g., the screw axis of a helix). The field~\(\psi\) possesses \textbf{axial \(C_n\) symmetry} (for an integer \(n \ge 1\)) if it remains invariant under rotation by \(\frac{2\pi}{n}\) radians about the \(x\)-axis. That is, for all \(\mathbf{r} \in \mathbb{R}^3\),
\begin{equation}\label{eq:cn-symmetry-def}
\psi(\mathbf{r}) = \psi\left(R_x\left(\frac{2\pi k}{n}\right)\mathbf{r}\right) \quad \text{for all } k \in \{0, 1, \dots, n-1\},
\end{equation}
where \(R_x(\theta)\) is the rotation matrix about the \(x\)-axis defined in~\eqref{eq:rotation-matrix}.
\end{definition}

Note that \(C_1\) symmetry is trivial, as it implies invariance under a rotation of \(2\pi\), which is true for any function. For \(n \ge 2\), \(C_n\) symmetry imposes a non-trivial constraint. If a structure possesses \(C_n\) symmetry, it also possesses \(C_k\) symmetry for any integer \(k\) that divides \(n\). As a convention, we refer to a structure as having a \(C_n\) symmetry if \(n\) is the largest integer for which the \(C_n\) symmetry holds.

\section{Problem Setup}\label{sec:problem-setup}

In Section~\ref{sec:preliminaries}, we introduced a mathematical framework for modeling helical structures and the projection process. We now formalize the problem addressed in this work.

We are given a collection of two-dimensional projection images. Each image corresponds to a segment of an unknown three-dimensional helical molecule, which is assumed to possess a known axial $C_n$ symmetry (Definition~\ref{def:cn-symmetry}). These projections are denoted $\Pi_B(t_i, \psi)$, where $\psi$ is the unknown helical structure, and $t_i \in \mathbb{R}$ is the unknown axial position of the segment from which the $i$-th image was obtained.

The goal is to recover the unknown\textbf{ projection angles} $\theta_i$ associated with each projection. As shown in Lemma~\ref{lemma:viewing-angle}, each projection $\Pi_B(t_i, \psi)$ is equivalent to a projection of a fixed reference segment, rotated about the screw axis by an angle $\theta_i$. Once these angles are determined, classical tomographic reconstruction techniques can be applied to reconstruct the 3D structure of the segment (and thus of the entire helix).

The data obtained using cryo-EM, which are essentially projections of biological structures, contains an inherent chirality ambiguity, which arises from the property shown in Lemma \ref{lemma:reflection-invariance}. Notice that by using the notation from~\eqref{eq:mirror-operator},~\eqref{eq:mirror-function}, and the definition of the rotation matrix \(R_x\left(\theta\right)\) in~\eqref{eq:rotation-matrix}, one can directly verify that
\[
MR_x\left(\theta\right)M = R_x\left(-\theta\right).
\]
Therefore, as a result of Lemma~\ref{lemma:reflection-invariance}, we have that
\begin{equation}\label{eq:chirality-ambiguity-x}
    P_{R_x(\theta)}\psi = P_{R_x(-\theta)}\psi^M.
\end{equation}
This means that valid projections of a helix or a segment are also valid projections of the mirror image of the same structure. Due to this inherent ambiguity, we aim to recover the projection angles only up to a global additive constant (reflecting the arbitrary choice of a reference segment) and a possible sign flip (reflecting the unresolved chirality).

Note that here we formulated the problem in terms of continuous helices. The more realistic discrete case will be dealt with in Section~\ref{sec:discrete-helix-angle-recovery}

\section{Prior work}\label{sec:prior-work}

Historically, the reconstruction of helical assemblies from cryo-EM images began with Fourier-Bessel methods~\cite{fourier-bessel-reconstruction,DEROSIER1968}, which leverage the observation that the Fourier transform of a perfectly helical structure consists of a set of layer lines indexed by Bessel functions. These approaches allow one to infer the helical symmetry parameters -- specifically, the rise and twist of a discrete helix or the pitch in the continuous case (see Definition~\ref{def:helix}), by matching observed layer-line patterns with theoretical predictions. While foundational, these methods are sensitive to imperfections in the specimen and to image noise. They also assume long, well-ordered helices~\cite{fourier-bessel-reconstruction}, which are increasingly rare in modern cryo-EM, where molecules are often short, flexible, or structurally heterogeneous~\cite{helical-reconstruction-again}.

A major difficulty with Fourier-Bessel analysis is that the power spectrum of a segment's projection may correspond to multiple valid combinations of rise and twist~\cite{helical-reconstruction-again}. These ambiguities are especially pronounced when the image contrast varies across the radius of the helix, or when the repeating components of the structure (referred to as subunits) are irregularly spaced along the filament. These issues reduce the reliability of symmetry estimation based on power spectrum analysis, particularly under high noise or when structural variations are present~\cite{RELION-helical-reconstruction}.

To address these challenges, the iterative helical real-space reconstruction (IHRSR) method was developed~\cite{IHRSR}. IHRSR refines both the symmetry parameters and the alignment of overlapping segments iteratively, improving robustness to disorder. However, it depends on an initial estimate of the symmetry, and poor initialization can lead to convergence to incorrect solutions~\cite{helical-reconstruction-again}. Software like RELION and cryoSPARC have implemented IHRSR-based pipelines with enhanced optimization techniques~\cite{RELION-helical-reconstruction, CryoSPARC-paper}, but these tools still assume that the approximate symmetry parameters are either known or can be scanned exhaustively.

More recent approaches integrate single-particle analysis strategies into helical reconstruction workflows~\cite{RELION-helical-reconstruction-full, RELION-helical-reconstruction}. In these methods, overlapping filament segments are treated as independent particles, enabling robust classification and averaging. This allows for improved handling of heterogeneity and has enabled near-atomic reconstructions in favorable cases. However, these methods still require the symmetry parameters, either supplied a priori or estimated iteratively, and remain sensitive to incorrect symmetry assumptions~\cite{symmetry-trap}.

Despite these advances, a core challenge persists: reliable \emph{a priori} determination of the helical symmetry parameters. In practice, rise and twist are often estimated through trial-and-error, low-resolution power spectra, or user expertise~\cite{helical-reconstruction-again, RELION-helical-reconstruction}. These estimates may be biased or ambiguous, and if the assumed symmetry is incorrect, the resulting 3D reconstruction, even if high-resolution, can be fundamentally flawed.

Our work addresses this central issue by proposing a framework that directly estimates the relative projection angles between helical segments from the data, without requiring prior knowledge of the helical parameters. In particular, we build upon prior successful attempts for recovering the relative orientations of one-dimensional projections of two-dimensional objects~\cite{graph_laplacian_2d_tomography, random_proj}. By avoiding reliance on symmetry parameters during the initial alignment, our method circumvents the ambiguity inherent in power spectrum-based approaches and mitigates the risk of convergence to incorrect reconstructions due to poor initializations.

\section{Spectral angle recovery}

In this section, we introduce the \textbf{SHREC} (\textbf{S}pectral \textbf{H}elical \textbf{REC}onstruction) algorithm, which is designed to estimate the projection angles of helical segments from their 2D projection images without prior knowledge of the helical symmetry parameters. Our method builds on the spectral embedding framework developed in~\cite{graph_laplacian_2d_tomography}, adapting it to the setting of 2D projections of 3D helical volumes.

\subsection{Spectral embedding}\label{subsec:spectral-embedding}
Spectral embedding techniques offer a principled method for representing complex, high-dimensional datasets in a lower-dimensional space that reflects the intrinsic geometry of the data. These methods are grounded in the spectral properties of the graph Laplacian, an operator constructed from pairwise similarities among data points. Although defined using only local proximity information, the graph Laplacian encapsulates global structural characteristics of the dataset. In this section, we introduce the required theory of graph Laplacians and spectral embedding.

Consider \( N \in \mathbb{N} \) data points  \( \{\mathbf{x}_i\}_{i=1}^N \subset \mathbb{R}^n \) lying on a smooth, \(d\)-dimensional manifold~$\mathcal{M} $ embedded in $\mathbb{R}^{n}$ (that is \(\mathcal{M} \subset \mathbb{R}^n\)). Note that the structure of the manifold is unknown, nor its dimension. The first step in the spectral embedding process is constructing a matrix \(\mathbf{W} \in \mathbb{R}^{N \times N}\) by
\begin{equation}\label{eq:weight-matrix}
    W_{ij} = K\left(\frac{\|\mathbf{x}_i - \mathbf{x}_j\|^2_2}{2\varepsilon}\right),
\end{equation}
where \(\|\cdot\|_2\) is the Euclidean norm in~$\mathbb{R}^{n}$, \(K: \mathbb{R}\rightarrow\mathbb{R}\) is a semi-positive kernel, and \(\varepsilon > 0\) is the bandwidth parameter. A common choice for the kernel function is \(K(x) = \exp(-x)\). We then normalize~$\mathbf{W}$ to be row-stochastic, by multiplying it by the inverse of the diagonal matrix \(\mathbf{D}\in\mathbb{R}^{N\times N}\) given by
\begin{equation}\label{eq:D-matrix}
    D_{ii} = \sum_{j=1}^NW_{ij}.
\end{equation}
Then, the graph Laplacian matrix \(\mathbf{L}\in\mathbb{R}^{N\times N}\) is defined as 
\begin{equation}\label{eq:laplacian-def}
    \mathbf{L} = \mathbf{I} - \mathbf{D}^{-1}\mathbf{W},
\end{equation}
where \(\mathbf{I}\) is the \(N\times N\) identity matrix.

The matrix \(\mathbf{D}^{-1}\mathbf{W}\) has a complete set of eigenvectors, \(\{\mathbf{v}_i\}_{i=0}^{N-1}\), with a corresponding set of non-negative set of eigenvalues \(\{\lambda_i\}_{i=0}^{N-1}\) such that \(1 = \lambda_0 \ge \lambda_1 \ge \cdots \ge \lambda_{N-1} \ge 0 \)~\cite{diffusion-maps-paper}. Since \(\mathbf{D}^{-1}\mathbf{W}\) is row-stochastic (i.e. its elements are nonnegative and its rows sum to~1), its first eigenvector, which corresponds to the eigenvalue \(\lambda_0 = 1\), is constant, that is \(\mathbf{v}_0 \propto (1,1,\cdots,1)^T\). The next \(k\) eigenvectors, with \(k \ll N\), are used as spectral embedding of the data into \(k\)-dimensional space via
\begin{equation}\label{eq:spectral embedding}
    \mathbf{x}_i \mapsto \left(\mathbf{v}_1(i), \mathbf{v}_2(i), \cdots, \mathbf{v}_k(i)\right),
\end{equation}
with \(\mathbf{v}_j(i)\) being the \(i\)'th entry of the \(j\)'th eigenvector. Notice that from~\eqref{eq:laplacian-def}, the graph Laplacian~$\mathbf{L}$ has the same eigenvectors as the matrix  \(\mathbf{D}^{-1}\mathbf{W}\), and its eigenvalues are \(\mu_i = 1 - \lambda_i\), with \(\lambda_i\) being the \(i\)'th eigenvalue of \(\mathbf{D}^{-1}\mathbf{W}\).

The importance of the embedding~\eqref{eq:spectral embedding} stems from the connection of the graph Laplacian to the Laplace–Beltrami operator on the underlying manifold \(\mathcal{M}\), which is defined as
\begin{equation}\label{eq:laplace–beltrami}
    \Delta_{\mathcal{M}} f = \frac{1}{\sqrt{|g|}} \sum_{i,j=1}^{d} \frac{\partial}{\partial x^i} \left( \sqrt{|g|} g^{ij} \frac{\partial f}{\partial x^j} \right),
\end{equation}
where \(g\) is the metric tensor of the manifold, $\abs{g}$ is the aboslute value of its determinant, and $g^{ij}$ are the components of the inverse metric tensor. Assuming the data points \(\{\mathbf{x}\}_{i=1}^N\) are uniformly distributed over the manifold, then 
for any smooth function \(f: \mathcal{M} \rightarrow \mathbb{R}\), as \(N \rightarrow \infty\) and \(\varepsilon \rightarrow 0\), it holds with high probability that~\cite{graph-laplacian-paper}
\begin{equation}\label{eq:laplacian-limit}
\frac{1}{\varepsilon} \sum_{j=1}^{N} L_{ij} f(\mathbf{x}_j) = \frac{1}{2} \Delta_{\mathcal{M}} f(\mathbf{x}_i) + O\left( \frac{1}{N^{1/2} \varepsilon^{1/2 + d/4}}, \varepsilon \right).
\end{equation}
Essentially,~\eqref{eq:laplacian-limit} asserts that the graph Laplacian provides a pointwise approximation of the Laplace–Beltrami operator on the manifold \(\mathcal{M}\). Importantly, this approximation is obtained without access to the manifold nor its dimension, but only to the sampled points.

Equation~\eqref{eq:laplacian-limit} can be generalized to the case of data points \(\{\mathbf{x}_i\}_{i=1}^N\) that are not uniformly distributed over the manifold. This is done by defining 
\begin{equation}
    \tilde{\mathbf{W}} = \mathbf{D}^{-1}\mathbf{W}\mathbf{D}^{-1},
\end{equation}
where \(\mathbf{D}\) is the diagonal matrix defined in~\eqref{eq:D-matrix}, computing the diagonal matrix~$\tilde{\mathbf{D}}$ of rows sum of~$\tilde{\mathbf{W}}$, and defining the density-invariant graph Laplacian as \(\tilde{\mathbf{L}} = \mathbf{I} - \tilde{\mathbf{D}}^{-1} \tilde{\mathbf{W}}\) (compare with~\eqref{eq:D-matrix} and \eqref{eq:laplacian-def}). When \(N\rightarrow\infty\) and \(\varepsilon \rightarrow 0\), the density-invariant graph Laplacian satisfies~\eqref{eq:laplacian-limit} (with~\(\tilde{\mathbf{L}}\) instead of~\(\mathbf{L}\)), independently of the density from which the points \(\{\mathbf{x}_i\}_{i=1}^N\) were sampled. 

We will focus on data restricted to a one-dimensional closed manifold (namely, a closed curve), since as shown later, this is the geometry of the data in our problem. Specifically, let \(\Gamma: \mathbb{R} \rightarrow \mathbb{R}^n\) be a closed curve of length~$l$ parametrized by arc length \(s\). The Laplace–Beltrami operator on this curve, as defined in~\eqref{eq:laplace–beltrami}, reduces to the second derivative with respect to arc length, namely,
\begin{equation}\label{eq:curve-laplacian}
    \Delta_\Gamma f(s) = f''(s),
\end{equation}
for any smooth function \(f: \Gamma \rightarrow \mathbb{R}\). Since the curve is closed, functions on \(\Gamma\) must satisfy the periodic boundary condition
\begin{equation}\label{eq:boundary-condition}
    f(0) = f(l).
\end{equation}
The first eigenfunction of \(\Delta_\Gamma\) satisfying the condition in~\eqref{eq:boundary-condition} is the constant function \(\phi_0(s) = 1\), corresponding to the eigenvalue \(\mu_0 = 0\). The remaining eigenfunctions are given by
\begin{equation}\label{eq: S1 eigenfunctions}
    \phi_m(s) = \begin{cases}
        \cos\left(\dfrac{2\pi}{l} \cdot \left\lceil \dfrac{m}{2} \right\rceil s\right) & \text{if } m \text{ is odd}, \\
        \sin\left(\dfrac{2\pi}{l} \cdot \left\lceil \dfrac{m}{2} \right\rceil s\right) & \text{if } m \text{ is even},
    \end{cases} \quad m \in \mathbb{N},
\end{equation}
with corresponding eigenvalues
\begin{equation}\label{eq:S1 eigenvalues}
    \mu_m = \left(\dfrac{2\pi}{l}\right)^2 \left\lceil \dfrac{m}{2} \right\rceil^2.
\end{equation}
Note that all eigenvalues \(\mu_m\) for \(m \geq 1\) have multiplicity two, except for the zero eigenvalue~\(\mu_0\), which is non-degenerate. Embedding the curve \(\Gamma \subset \mathbb{R}^n\) into \(\mathbb{R}^2\) using the first two non-trivial eigenfunctions therefore results in a circle in the plane. Explicitly, the embedding map is
\begin{equation}\label{eq:circle-embedding}
    \Gamma(s) \mapsto \left(\cos\left(\dfrac{2\pi s}{l}  \right), \sin\left(\dfrac{2\pi s}{l} \right) \right)\quad s \in [0, l).
\end{equation}

Combining~\eqref{eq:circle-embedding} with convergence results for the eigenvectors of the graph Laplacian~\cite{cheng2022eigen}, we conclude that the spectral embedding of data sampled from a one-dimensional closed manifold~$\Gamma$ is approximately a circle in~$\mathbb{R}^2$. Specifically, if each data point is of the form $\mathbf{x}_i = \Gamma(s_i)$, then the spectral embedding~\eqref{eq:spectral embedding} reads as 
\begin{equation}\label{eq:approx-circle}
\mathbf{x}_i \mapsto \left(\mathbf{v}_1(i), \mathbf{v}_2(i)\right) \approx \left(\cos\left(\dfrac{2\pi s_i}{l}  \right), \sin\left(\dfrac{2\pi s_i}{l} \right) \right).
\end{equation}
As a result of the multiplicity of the eigenvalues (see~\eqref{eq: S1 eigenfunctions} and~\eqref{eq:S1 eigenvalues}),~\eqref{eq:approx-circle} holds up to a rotation (determining the point on the curve corresponding to $s=0$) and reflection (the direction of growing~$s$). See~\cite{graph_laplacian_2d_tomography} for further details.

Since~\eqref{eq:circle-embedding} and~\eqref{eq:approx-circle} describe a circle, it is natural to introduce an angular parametrization, which reflects the intrinsic periodicity of the underlying manifold and offers a convenient one-dimensional coordinate system. Define the angular coordinate $\varphi : [0, l) \rightarrow [0, 2\pi)$ by
\begin{equation}\label{eq:angular-param}
\varphi(s) = \frac{2\pi}{l}s.
\end{equation}
This mapping linearly transforms the arc-length parameter $s$ to an angular variable $\varphi$, such that a full traversal of the curve $\Gamma$ corresponds to a full revolution around the unit circle. Under this parametrization, the embedding~\eqref{eq:circle-embedding} becomes
\begin{equation}
\Gamma(\varphi) \mapsto \left(\cos(\varphi), \sin(\varphi)\right),
\end{equation}
with $\varphi \in [0, 2\pi)$.

\subsection{Projection angle recovery}\label{subsec:projection-angle-recovery}

Having established the theoretical foundation for recovering the intrinsic geometry of a one-dimensional closed manifold from high-dimensional data via spectral embedding, we now connect these insights to the problem of projection angle estimation in cryo-EM.
% \begin{definition}\label{def:projection-angle-function}
% Let \(\psi\) be a continuous helix with pitch \(P \in \mathbb{R}\setminus \{0\}\). Define the map \(\Theta:\left\{\Pi_B(t, \psi) : t \in \mathbb{R}\right\} \to S^1\) such that for any projection \(\Pi_B(t, \psi)\) of a segment taken at axial position \(t\),
% \begin{equation}
%     \Theta(\Pi_B(t, \psi)) = \frac{2\pi}{P}t \pmod{2\pi}.
% \end{equation}
% This function assigns to each segment's projection its corresponding projection angle relative to the reference segment at \(t_0 = 0\).
% \end{definition}
First, we derive the manifold structure of a set of projection images. Lemma~\ref{lemma:projection-differentiable} below establishes the differentiability of the map $\Phi: S^1 \to L^2(\Omega_B)$ that maps rotation angles to projected images. This differentiability is a prerequisite for proving that projection images have a manifold structure, as established in Theorem~\ref{theorem:manifold-structure}. The proofs of Lemma~\ref{lemma:projection-differentiable} and Theorem~\ref{theorem:manifold-structure} are given in Appendix~\ref{app:additional-proofs}.

\begin{lemma}\label{lemma:projection-differentiable}
Let \(B > 0\), and let 
\(
f \;\in\; C^1\bigl(\bigl[-\tfrac{B}{2},\tfrac{B}{2}\bigr]^3\bigr),
\)
i.e., \(f\) is continuously differentiable on the cube \(\bigl[-\tfrac{B}{2},\tfrac{B}{2}\bigr]^3\). Assume that the support of \(f\) contained in a cylinder with radius of \(\frac{B}{2}\), that is
\begin{equation}\label{eq:f-cyl-support}
\mathrm{supp}(f) \subseteq \left\{(x,y,z)\in \mathbb{R}^3 : \sqrt{y^2 + z^2} \le \frac{B}{2} \wedge |x| < \frac{B}{2} \right\}.
\end{equation}
For each \(\theta\in S^1\), let 
$R_x(\theta)\colon \mathbb{R}^3 \;\longrightarrow\; \mathbb{R}^3$
denote the rotation about the \(x\)‐axis by angle~\(\theta\), given in~\eqref{eq:rotation-matrix}.  Define 
\[
\Phi\colon S^1\;\longrightarrow\;L^2\Bigl(\bigl[-\tfrac{B}{2},\tfrac{B}{2}\bigr]^2\Bigr)
\]
by
\begin{equation}\label{eq:Phi-definition}
\Phi(\theta)(x,y)\;=\;(P_{R_x(\theta)}f)(x,y).
\end{equation}
Then, $\Phi$ is continuously differentiable. Specifically, 
\begin{equation}\label{eq:Phi‐derivative}
\Phi'(\theta)(x,y)
=\int_{-\tfrac{B}{2}}^{\tfrac{B}{2}}
\nabla f\bigl(R_x(\theta)(x,y,z)^T\bigr)\,\cdot\,\bigl(\dot R_x(\theta)(x,y,z)^T\bigr)\,dz,
\end{equation}
where \(\dot R_x(\theta) = \tfrac{d}{d\theta}\,R_x(\theta)\).
\end{lemma}

\begin{theorem}\label{theorem:manifold-structure}
Let \(B > 0\), and let 
\(
f \;\in\; C^1\bigl(\bigl[-\tfrac{B}{2},\tfrac{B}{2}\bigr]^3\bigr)\;\cap\;L^2\bigl(\bigl[-\tfrac{B}{2},\tfrac{B}{2}\bigr]^3\bigr).
\)
Assume that~\(f\) satisfies~\eqref{eq:f-cyl-support} and is invariant under the cyclic group $C_{n}$ for some $n\in\mathbb{N}$, and has no cyclic symmetry of higher order and no continuous rotational symmetry. Assume further that the corresponding \(\Phi\) (as defined in~\eqref{eq:Phi-definition}) satisfies the following
\begin{enumerate}
  \item[(i)] (\emph{Injectivity modulo $C_n$}) 
    \(\Phi(\theta_1)=\Phi(\theta_2)\)
    if and only if \(\theta_2-\theta_1\equiv \tfrac{2\pi k}{n}\pmod{2\pi}\).
  \item[(ii)] (\emph{Nondegeneracy}) 
    \(\Phi'(\theta)\neq 0\) in \(L^2(\bigl[-\tfrac{B}{2},\tfrac{B}{2}\bigr]^2)\) for every \(\theta\in S^1\).
\end{enumerate}
Then, the image
\[
\mathcal M \;=\;\Phi\bigl(S^1\bigr)
\;=\;\bigl\{P_{R_x(\theta)}f:\theta\in S^1\bigr\}
\]
is a closed, embedded, one‐dimensional \(C^1\)–submanifold of \(L^2(\bigl[-\tfrac{B}{2},\tfrac{B}{2}\bigr]^2)\).  In particular,~\(\mathcal M\) is diffeomorphic to the circle \(S^1\).
\end{theorem}

The preceding theorem provides the theoretical basis for recovering projection angles from projections of helical segments. As we showed in Lemma~\ref{lemma:viewing-angle}, a set of projected helical segments is equivalent to a set of projections of a reference helical segment from different angles. The following lemma connects this observation with Theorem~\ref{theorem:manifold-structure} by establishing that helical segment projections form a one-dimensional manifold in $L^2(\mathbb{R}^2)$. This manifold structure provides the analytical framework for assigning a projection angle to each helical segment via a diffeomorphism that parametrizes the manifold by the angular variable.

\begin{theorem}\label{lemma:helical-projections-manifold}
    Let \(\psi: \mathbb{R}^3 \rightarrow \mathbb{R}\) be a smooth \(C^\infty\) helix with \(C_n\) symmetry and support contained in a cylinder with radius of \(\frac{B}{2}\), that is
    \[
    \mathrm{supp}(\psi) \subseteq \left\{(x,y,z)\in \mathbb{R}^3 : \sqrt{y^2 + z^2} \le \frac{B}{2} \right\}.
    \]
    Assume that all segments \(S_B(\psi,t)\)
    of \(\psi\) (see Definition~\ref{def:helical-segment}) 
    satisfy conditions \textit{(i)} and \textit{(ii)} from Theorem~\ref{theorem:manifold-structure} (injectivity modulo $C_n$; nondegeneracy of the derivative). Then, the following hold
    \begin{enumerate}
        \item The set of segment projections \(\{\Pi_B(t,\psi) : t \in \mathbb{R}\}\) is a closed, one‐dimensional submanifold of \(L^2(\bigl[-\tfrac{B}{2},\tfrac{B}{2}\bigr]^2)\), diffeomorphic to the circle \(S^1\).
        \item  Define \(\Gamma_\psi: S^1 \to \{\Pi_B(t,\psi) : t \in \mathbb{R}\}\) by 
        \begin{equation}\label{eq:gamma-definition}
            \Gamma_\psi(\varphi) = P_{R_x(\varphi/n)}S_B(0,\psi)
        \end{equation}
        for all \(\varphi \in [0, 2\pi)\), where $S_{B}$ is given in Definition~\ref{def:helical-segment}. Then \(\Gamma_\psi\) is well defined and a diffeomorphism.
    \end{enumerate}
\end{theorem}

\begin{proof}[Proof of Theorem \ref{lemma:helical-projections-manifold}]
We prove this result by first showing that the set of all helical segment projections $\{\Pi_B(t,\psi) : t \in \mathbb{R}\}$ coincides with the set of projections of a fixed reference segment viewed from all angles. We then apply Theorem~\ref{theorem:manifold-structure} to establish that this set forms a one-dimensional manifold, and explicitly construct the diffeomorphism~$\Gamma_\psi$ that parametrizes it.

\textbf{Step 1: Smoothness of segments.}
Since $\psi \in C^\infty(\mathbb{R}^3)$ is a smooth helix, for any fixed $t_0 \in \mathbb{R}$, the segment $S_B(t_0, \psi)$ defined by
\[
S_B(t_0, \psi)(\mathbf{r}) = \psi(\mathbf{r} - t_0\hat{\mathbf{x}})
\]
for $\mathbf{r} \in Q_B = [-\frac{B}{2}, \frac{B}{2}]^3$, inherits the smoothness of $\psi$. Specifically, since translation is a smooth operation and $\psi$ is smooth, the composition $\mathbf{r} \mapsto \psi(\mathbf{r} - t_0\hat{\mathbf{x}})$ is also smooth on~$Q_B$. Therefore, $S_B(t_0, \psi) \in C^\infty(Q_B)$.

\textbf{Step 2: Establishing the relationship between segment projections and rotations.}
Fix a reference segment at $t_0 = 0$ and denote $f = S_B(0, \psi)$. By Lemma~\ref{lemma:viewing-angle}, for any $t \in \mathbb{R}$, the projection $\Pi_B(t, \psi)$ can be expressed as
\[
\Pi_B(t, \psi) = P_{R_x(\theta_t)} S_B(0, \psi) = P_{R_x(\theta_t)} f,
\]
where $\theta_t = \frac{2\pi}{P}t$ is the angle between segments at positions $0$ and $t$.

Since $t$ ranges over all of $\mathbb{R}$ and the map $t \mapsto \frac{2\pi}{P}t \pmod{2\pi}$ covers all of $[0, 2\pi)$, we have
\[
\{\Pi_B(t, \psi) : t \in \mathbb{R}\} = \{P_{R_x(\theta)} f : \theta \in [0, 2\pi)\}.
\]

\textbf{Step 3: Verifying the conditions of Theorem~\ref{theorem:manifold-structure}.}
Let $\Omega_B = [-\frac{B}{2}, \frac{B}{2}]^2$. Recall that $\Phi: S^1 \rightarrow L^2(\Omega_B)$ is defined by $\Phi(\theta) = P_{R_x(\theta)} f$.

(a) \emph{Verifying the support condition:} Since the support of $\psi$ is contained in a cylinder of radius~$\frac{B}{2}$ around the $x$-axis, we have
\[
\mathrm{supp}(f) = \mathrm{supp}(S_B(0, \psi)) \subseteq \left\{(x,y,z) \in \mathbb{R}^3 : \sqrt{y^2 + z^2} \le \frac{B}{2} \wedge |x| < \frac{B}{2}\right\}.
\]

(b) \emph{Verifying smoothness:} Since $f \in C^\infty(Q_B)$ (from Step 1), we have $f \in C^1(Q_B) \cap L^2(Q_B)$.

(c) \emph{Verifying $C_n$ symmetry:} The $C_n$ symmetry of $\psi$ is inherited by any segment, so $f$ has $C_n$ symmetry.

By the theorem's hypothesis, $f$ is such that its corresponding \(\Phi\) (as defined in~\eqref{eq:Phi-definition}) satisfies conditions (i) and (ii) of Theorem~\ref{theorem:manifold-structure}. Therefore, by Theorem~\ref{theorem:manifold-structure}, $\mathrm{im}(\Phi) = \{P_{R_x(\theta)} f : \theta \in S^1\}$ is a closed, embedded, one-dimensional $C^1$-submanifold of $L^2(\Omega_B)$, diffeomorphic to~$S^1$. From Step~2, $\{\Pi_B(t, \psi) : t \in \mathbb{R}\} = \mathrm{im}(\Phi)$, which proves assertion~1 of the lemma.

\textbf{Step 4: Proving that $\Gamma_\psi$ is a diffeomorphism.}
Recall that $\Gamma_\psi: S^1 \rightarrow \{\Pi_B(t, \psi) : t \in \mathbb{R}\}$ is defined by
\[
\Gamma_\psi(\varphi) = P_{R_x(\varphi/n)} S_B(0, \psi) = \Phi(\varphi/n).
\]

(a) \emph{Well-definedness:} Since $\Phi$ has period $\frac{2\pi}{n}$ due to the $C_n$ symmetry (i.e., $\Phi(\theta + \frac{2\pi}{n}) = \Phi(\theta)$), the map $\varphi \mapsto \Phi(\varphi/n)$ has period $2\pi$. Thus $\Gamma_\psi$ is well-defined on $S^1$.

(b) \emph{Injectivity:} Suppose $\Gamma_\psi(\varphi_1) = \Gamma_\psi(\varphi_2)$. Then $P_{R_x(\varphi_1/n)} f = P_{R_x(\varphi_2/n)} f$. By condition (i) of Theorem~\ref{theorem:manifold-structure}, this implies
\[
\frac{\varphi_2 - \varphi_1}{n} \equiv \frac{2\pi k}{n} \pmod{2\pi}
\]
for some $k \in \{0, 1, \ldots, n-1\}$. Therefore, $\varphi_2 - \varphi_1 \equiv 2\pi k \pmod{2\pi}$, which means $\varphi_1 = \varphi_2$ on $S^1$. Thus $\Gamma_\psi$ is injective.

(c) \emph{Surjectivity:} For any $\Pi_B(t, \psi) \in \{\Pi_B(t, \psi) : t \in \mathbb{R}\}$, by Step 2 there exists $\theta \in [0, 2\pi)$ such that $\Pi_B(t, \psi) = P_{R_x(\theta)} f$. Setting $\varphi = n\theta \pmod{2\pi}$, we have
\[
\Gamma_\psi(\varphi) = P_{R_x(\varphi/n)} f = P_{R_x(\theta)} f = \Pi_B(t, \psi).
\]
Therefore,~$\Gamma_\psi$ is surjective.

(d) \emph{Smoothness with non-vanishing derivative:} The map $\Gamma_\psi(\varphi) = \Phi(\varphi/n)$ is the composition of the smooth map $\varphi \mapsto \varphi/n$ with the $C^1$ map $\Phi$ (by Lemma~\ref{lemma:projection-differentiable}). Therefore, $\Gamma_\psi$ is $C^1$.

The derivative is given by
\[
\Gamma_\psi'(\varphi) = \frac{1}{n} \Phi'(\varphi/n).
\]
By condition (ii) of Theorem~\ref{theorem:manifold-structure}, $\Phi'(\theta) \neq 0$ for all $\theta \in S^1$. Therefore, $\Gamma_\psi'(\varphi) \neq 0$ for all $\varphi \in S^1$.

Since $\Gamma_\psi: S^1 \rightarrow \{\Pi_B(t, \psi) : t \in \mathbb{R}\}$ is a bijective $C^1$ map between one-dimensional manifolds with non-vanishing derivative everywhere, it is a diffeomorphism. This completes the proof of assertion 2 of the lemma.
\end{proof}

Theorem \ref{lemma:helical-projections-manifold} shows that a set of projections of helical segments \(\{\Pi_B(t,\psi)\}_{t\in\mathbb{R}}\) embeds as a smooth, closed one-dimensional submanifold in \(L^2(\bigl[-\tfrac{B}{2},\tfrac{B}{2}\bigr]^2)\) that is diffeomorphic to \(S^1\). 

We now apply the results from Lemma~\ref{lemma:helical-projections-manifold} and Theorem~\ref{theorem:manifold-structure} to the problem of estimating projection angles from observed cryo-EM images. Given a set of $N$ observed projection images $\{\Pi_j\}_{j=1}^N$, where each $\Pi_j = \Pi_B(t_j, \psi)$ for some unknown axial position $t_j$, it follows from Lemma~\ref{lemma:helical-projections-manifold} that this set lies on a one-dimensional submanifold of \(L^2(\bigl[-\tfrac{B}{2},\tfrac{B}{2}\bigr]^2)\) diffeomorphic to \(S^1\). This connects naturally to the spectral embedding framework developed in Section~\ref{subsec:spectral-embedding}. By applying spectral embedding to these projections, we obtain coordinates \((\mathbf{v}_1(j), \mathbf{v}_2(j))\) for each image, which according to Section~\ref{subsec:spectral-embedding}, approximate points on a circle. From these coordinates, we can compute angular positions
\begin{equation}\label{eq:embedding-angle}
    \varphi_j = \operatorname{atan2}\left(\mathbf{v}_2(j), \mathbf{v}_1(j)\right).
\end{equation}
As established in~\eqref{eq:approx-circle} and~\eqref{eq:angular-param}, these embedding-derived angles \(\varphi_j\) approximate, up to a global rotation \(\phi_0 \in [0, 2\pi)\) and a possible reflection, the arc-length parametrization of the underlying manifold, namely,
\begin{equation}\label{eq:phi-s-relation}
    \varphi_j \approx \pm \frac{2\pi s_j}{l} + \phi_0,
\end{equation}
where \(s_j\) is the arc-length parameter corresponding to projection \(\Pi_j\), and \(l\) is the total length of the manifold.

Our objective is to connect these embedding angles \(\varphi_j\) to the actual projection angles~\(\theta_j\) that define the viewing direction of each helical segment. Let us choose a reference segment \(S_0 = S_B(0, \psi)\) at position \(t_0 = 0\). According to Lemma~\ref{lemma:viewing-angle}, each projection \(\Pi_j = \Pi_B(t_j, \psi)\) is equivalent to viewing this reference segment from angle \(\theta_j = \frac{2\pi}{P}t_j\). Due to the $C_n$ symmetry of the helix, the manifold $\mathcal{M}_f = \{P_{R_x(\theta)}f : \theta \in S^1\}$ (where $f = S_B(0,\psi)$ is the reference segment) is fully traced as the projection angle $\theta = \frac{2\pi}{P}t$ varies over any interval of length $2\pi/n$ (e.g., $[0, 2\pi/n)$). Consequently, as the projection angle $\theta$ completes a full $2\pi$ revolution, the manifold $\mathcal{M}_f$ is traversed $n$ times. Under the assumption of approximately constant speed parametrization (see ~\eqref{eq:constant-speed-assumption} below), the arc-length parameter \(s\) and the projection angle \(\theta\) are related by \begin{equation}\label{eq:s-theta-relation} s = \frac{l \cdot n \cdot \theta}{2\pi}, \end{equation} where the factor \(n\) accounts for the \(C_n\) symmetry.

To establish a direct relationship between the embedding angles~\(\varphi_j\) and the projection angles~\(\theta_j\), we introduce an assumption regarding how the manifold \(\mathcal{M}_f\) is parametrized by the projection angle. Specifically, we assume that the parametrization of \(\mathcal{M}_f\) by \(\theta \in [0, \frac{2\pi}{n})\) has approximately constant speed with respect to the \(L^2\)-norm, that is
\begin{equation}\label{eq:constant-speed-assumption}
    \left\|\frac{d}{d\theta} \left(P_{R_x(\theta)} S_0\right) \right\|_{L^2(\mathbb{R}^2)} \approx C,
\end{equation}
for some constant \(C > 0\) for all \(\theta \in [0, \frac{2\pi}{n})\). The assumption of~\eqref{eq:constant-speed-assumption} implies that arc-length along the manifold \(\mathcal{M}_f\) is approximately proportional to the change in projection angle~\(\theta\). While not universally guaranteed, this is often a reasonable approximation for helical structures where the projected mass or structural features change at a relatively consistent rate as the viewing angle~\(\theta\) varies over~\([0, 2\pi/n)\). Such behavior can be expected if the molecule's principal features are distributed somewhat uniformly around its helical axis, rather than being concentrated at specific azimuthal angles.

Combining Equations \eqref{eq:phi-s-relation} and \eqref{eq:s-theta-relation}, we find that the relationship between the embedding angle~$\varphi_j$ and the true projection angle~$\theta_j$ is
\begin{equation}\label{eq:phi-theta-final-relation}
    \varphi_j \approx \pm n\theta_j + \phi_0 \pmod{2\pi}.
\end{equation}
This equation signifies that the angles $\varphi_j$ recovered from spectral embedding are, up to a global rotation $\phi_0$ and a sign, $n$ times the true projection angles $\theta_j$, modulo $2\pi$. Therefore, the projection angles $\theta_j$ can be estimated as $\theta_j \approx \pm \frac{1}{n}(\varphi_j - \phi_0) \pmod{2\pi/n}$.

Algorithm~\ref{alg:angle-extraction-continuous} below describes the angle extraction process.

\begin{algorithm}[H]
\caption{SHREC: Spectral Angle Recovery for Continuous Helices}
\label{alg:angle-extraction-continuous}
\begin{algorithmic}[1]
\REQUIRE A set of $N$ helical segment projection images $\{\Pi_i\}_{i=1}^N \subset L^2(\mathbb{R}^2)$.
\REQUIRE The order of axial symmetry $n \in \mathbb{N}$ of the helix (cf. Definition~\ref{def:cn-symmetry}).
\ENSURE Estimated relative projection angles $\{\theta_i\}_{i=1}^N$, unique up to a global shift and a sign flip.

\COMMENT{Spectral Embedding}
\STATE Compute pairwise $L^2$-distances $d_{ij} = \|\Pi_i - \Pi_j\|_{L^2(\mathbb{R}^2)}$.
\STATE Construct a similarity matrix $\mathbf{W}$ with $W_{ij} = \exp(-d_{ij}^2 / (2\varepsilon))$ (cf.~\eqref{eq:weight-matrix}).
\STATE Construct the density-invariant graph Laplacian $\tilde{\mathbf{L}} = \mathbf{I} - \tilde{\mathbf{D}}^{-1}\tilde{\mathbf{W}}$, where $\tilde{\mathbf{W}} = \mathbf{D}^{-1}\mathbf{W}\mathbf{D}^{-1}$ and $\tilde{\mathbf{D}}$ is the diagonal matrix of row sums of $\tilde{\mathbf{W}}$ (cf. Section~\ref{subsec:spectral-embedding}).

\STATE Compute the eigenvectors $\mathbf{u}_1, \mathbf{u}_2$ of $\mathbf{S} = \tilde{\mathbf{D}}^{-1/2}\tilde{\mathbf{W}}\tilde{\mathbf{D}}^{-1/2}$ corresponding to the second and third largest eigenvalues, and set $\tilde{\mathbf{v}}_i = \tilde{\mathbf{D}}^{1/2}\mathbf{u}_i$ for $i = 1, 2$.

\COMMENT{Angle Extraction}
\STATE For each $\Pi_i$, compute its embedding angle $\varphi_i = \operatorname{atan2}(\tilde{\mathbf{v}}_2(i), \tilde{\mathbf{v}}_1(i))$ (cf.~\eqref{eq:embedding-angle}).
\STATE Set the estimated projection angles
    $\theta_i = \varphi_i / n$.
\RETURN Estimated projection angles $\{\theta_i\}_{i=1}^N$.
\end{algorithmic}
\end{algorithm}

\subsection{Discrete helices}\label{sec:discrete-helix-angle-recovery}
In the preceding sections, we developed a framework for angle recovery based on an idealized model of a continuous helix. While this model provides a clean mathematical foundation, real biological helical assemblies (see Definition~\ref{def:helix}) are composed of discrete, repeating subunits. This discrete nature breaks the exact correspondence between translation and rotation established in Lemma~\ref{lemma:translation-rotation-correspondence}. Specifically, for a discrete helix, an arbitrary translation along the screw axis cannot be perfectly compensated by a rotation about that axis, as shown below.

This section extends our analysis to the more realistic case of discrete helices. We will show that while the projections of discrete helical segments do not lie exactly on a one-dimensional manifold, they remain close to the ideal manifold defined by a corresponding continuous structure. This proximity allows us to justify the application of the spectral angle recovery method to helical segments of discrete helices, treating the deviation from the ideal manifold as a form of bounded noise.

We begin by formalizing the relationship between translation and rotation for a discrete helix.

\begin{lemma}\label{lemma:discrete-symmetry-equivalence}
Let $\psi$ be a discrete helix with rise $\Delta x$ and twist $\Delta\theta$.
Then, for any $t \in \mathbb{R}$ and any $\mathbf{r} \in \mathbb{R}^3$, there exist \(\theta_t \in [0, 2\pi)\) and \(s_t \in [-\frac{\Delta x}{2}, \frac{\Delta x}{2})\) such that
\begin{equation}\label{eq:discrete-symmetry-form}
\psi\left(\mathbf{r} + t\hat{\mathbf{x}}\right) = \psi\left(R_{x}\left(\theta_t\right)\mathbf{r} + s_t\hat{\mathbf{x}}\right).
\end{equation}
\end{lemma}

Lemma~\ref{lemma:discrete-symmetry-equivalence}  shows that unlike the continuous case (Lemma~\ref{lemma:translation-rotation-correspondence}), a translation in the discrete case decomposes into a discrete symmetry operation and a small residual shift. Specifically, it shows that shift of a discrete helix along its screw axis, is equivalent to rotation around the screw axis and a small residual shift $s_t$, where $|s_t| \le \rise/2$. When we consider the projection images, this residual shift $s_t$ introduces a deviation from the ideal manifold~$\mathcal{M}$ that we characterized in Theorem~\ref{theorem:manifold-structure}. In Theorem \ref{theorem:manifold-distance-bound} below, we provide a bound on this deviation. To quantify it, we recall our notation
\begin{equation}\label{eq:box-sets}
Q_{B}= \Bigl[-\tfrac{B}{2},\tfrac{B}{2}\Bigr]^{3},\qquad   
\Omega_{B}= \Bigl[-\tfrac{B}{2},\tfrac{B}{2}\Bigr]^{2},
\end{equation}
and define for \(f \in C^1(Q_B)\),
\begin{equation}
M_x(f) \;=\; \sup_{Q_{B}}\left|\frac{\partial f}{\partial x}\right|.
\end{equation}

\begin{theorem}\label{theorem:manifold-distance-bound}
Let $\psi$ be a discrete helix with rise $\Delta x$ and twist $\Delta\theta$, such that \(\psi \in C^1(\mathbb{R}^3)\) (that is, continuously differentiable) and \(\left|\frac{\partial \psi}{\partial x}\right| < \infty\).
Define
$$
\Pi(t)=\Pi_{B}(t,\psi)=P_{I}S_{B}(t,\psi),
\qquad t\in\mathbb R,
$$
where $P_{I}$ is the projection operator defined in~\eqref{eq:ideal-projection-model} and $S_{B}$ is the segment operator from Definition~\ref{def:helical-segment}. Let $$ \mathcal M_{\text{ideal}} = \bigl\{P_{R_x(\theta)}S_{B}(0,\psi):\theta\in[0,2\pi)\bigr\} $$ denote the ideal projection manifold of the reference helical segment at \(t=0\).
Then
\begin{equation}\label{eq:theorem-bound}
d_{L^{2}(\Omega_{B})}\!\bigl(\Pi(t),\mathcal M_{\text{ideal}}\bigr)
\le\;
\frac{1}{2}\rise\,M_{x}(\psi)\,B^{\frac{3}{2}},
\quad t\in\mathbb R,
\end{equation}
where 
\[ d_{L^{2}(\Omega_{B})}\!\bigl(\Pi(t),\mathcal M_{\text{ideal}}\bigr)=\inf_{\theta\in[0,2\pi)} \bigl\|\Pi(t)-P_{R_x(\theta)}S_{B}(0,\psi)\bigr\|_{L^{2}(\Omega_{B})}. 
\]
\end{theorem}

Theorem~\ref{theorem:manifold-distance-bound} provides the key theoretical justification for applying our spectral angle recovery method to discrete helices. It establishes that every projection image from a discrete helix lies within a bounded distance from the ideal manifold $\mathcal{M}_{\text{ideal}}$. This bound is directly proportional to the rise~$\rise$ and to~$M_x(\psi)$, which is the smoothness of the structure along the helical axis. 

This result implies that for a discrete helix with a sufficiently small rise that is sufficiently smooth, the set of projection images forms a ``thickened" version of the one-dimensional manifold $\mathcal{M}_{\text{ideal}}$. As long as this ``thickness" (quantified by the bound in~\eqref{eq:theorem-bound}) is small, the spectral embedding method described in Section~\ref{subsec:spectral-embedding} can reliably recover the underlying circular topology of the manifold and, consequently, provide accurate estimates of the projection angles. Therefore, Algorithm~\ref{alg:angle-extraction-continuous} can be applied directly to data from discrete helices, with the understanding that the recovered angles are approximations whose accuracy is governed by the helical parameters and the smoothness of the underlying structure. We demonstrate the validity of this approach experimentally in Section~\ref{sec:results}.

\section{Implementation details}\label{sec:implementation-details}

This section details the practical application of the SHREC algorithm within a complete helical reconstruction pipeline. The pipeline is designed to integrate with the widely used RELION~\cite{RELION-paper} software suite, leveraging its robust tools for data handling and refinement, while incorporating the  SHREC algorithm for angle recovery. The entire pipeline can be divided into four primary stages: data preprocessing, signal denoising, spectral angle recovery, and final 3D reconstruction. The input to the reconstruction pipeline is assumed to be a set of unaligned movie frames of the helical dataset. 

\subsection{Preprocessing}
The initial stage of the helical reconstruction pipeline involves standard cryo-EM data processing steps performed within RELION to produce a clean and consistently oriented set of helical segment images. These steps include:
\begin{enumerate}
    \item \textbf{Micrograph processing:} Raw movie frames are first subjected to beam-induced motion correction and per-micrograph contrast transfer function (CTF) estimation using the standard RELION procedures.
    \item \textbf{Helical segment extraction:} Helical filaments are manually or semi-automatically picked using RELION's helical tools. A stack of overlapping 2D images, which are the helical segment projections, is then extracted along these filaments.
    \item \textbf{2D Classification and alignment:} The extracted particle stack is subjected to reference-free 2D classification with the helical option active (\texttt{relion\_class\_2d}). This serves two functions:
    \begin{itemize}
        \item Data pruning: identifying and removing particles corresponding to noise, ice artifacts, or poorly formed helical segments.
        \item In-plane alignment: The resulting high-quality class averages reveal the orientation of the helical axis within the 2D image plane. This information is used to determine the in-plane rotational and translational parameters required to align all segments.
    \end{itemize}
\end{enumerate}

Following 2D classification, all selected helical segment projections can be aligned such that their helical axis is centered and oriented horizontally. The output of this stage is a consistently aligned stack of 2D images, which serves as the input to the SHREC denoising module.

\subsection{Denoising}
Cryo-EM images are characterized by a low signal-to-noise ratio (SNR). The spectral embedding method detailed in Section~\ref{subsec:spectral-embedding} relies on computing pairwise distances between images. In the presence of high noise levels, these distances are dominated by random fluctuations rather than by true structural differences, which can obscure the underlying  manifold structure of the data. To mitigate this, we employ a denoising step based on Wiener filtering prior to spectral embedding. The Wiener filter is the optimal linear filter for enhancing a signal corrupted by additive noise, provided that the statistical properties of the signal and noise are known or can be estimated.

Let the set of observed, aligned helical segment images be denoted by $\{y_i(\mathbf{x})\}_{i=1}^N$, where $\mathbf{x} \in \mathbb{R}^2$ represents the spatial coordinates. We model each image as the sum of a true, noise-free signal $s_i(\mathbf{x})$ and an additive noise component $n_i(\mathbf{x})$, that is
\begin{equation}
    y_i(\mathbf{x}) = s_i(\mathbf{x}) + n_i(\mathbf{x}).
\end{equation}
In the Fourier domain, this relationship becomes
\begin{equation}
    Y_i(\mathbf{f}) = S_i(\mathbf{f}) + N_i(\mathbf{f}),
\end{equation}
where $Y_i$, $S_i$, and $N_i$ are the Fourier transforms of $y_i$, $s_i$, and $n_i$, respectively, and $\mathbf{f} \in \mathbb{R}^2$ is the spatial frequency vector. We treat the signal and noise as realizations of zero-mean, stationary random processes, and we assume they are uncorrelated.

The Wiener filter, $G(\mathbf{f})$, is designed to produce an estimate $\hat{S}_i(\mathbf{f}) = G(\mathbf{f})Y_i(\mathbf{f})$ that minimizes the mean squared error $\mathbb{E}[|\hat{S}_i(\mathbf{f}) - S_i(\mathbf{f})|^2]$. The optimal filter is given by (see Appendix~\ref{app:wiener-filter} for a full derivation)
\begin{equation}\label{eq:wiener-filter-def}
    G(\mathbf{f}) = \frac{P_{SS}(\mathbf{f})}{P_{SS}(\mathbf{f}) + P_{NN}(\mathbf{f})},
\end{equation}
where $P_{SS}(\mathbf{f})$ and $P_{NN}(\mathbf{f})$ are the power spectral densities (PSDs) of the signal and noise, respectively. The PSD of a random process $X$ is defined as $P_{XX}(\mathbf{f}) = \mathbb{E}[|X(\mathbf{f})|^2]$, where the expectation is taken over the ensemble of all possible realizations.

In practice, the true signal and noise PSDs are unknown and must be estimated from the noisy observations $\{Y_i\}_{i=1}^N$. We employ a method based on Principal Component Analysis (PCA), which leverages the distinct statistical properties of the signal and noise. It's underlying assumption is that the coherent signal is concentrated in the low-order (high-variance) principal components (PCs), while the noise, being less correlated between images, is distributed more uniformly across all components.  To estimate the noise PSD, which we denote as $\hat{P}_{NN}(\mathbf{f})$, we average the power spectra of images reconstructed from higher-order PCs. The distinction between the high-order and low-order PCs is done by a simple "elbow method" heuristic. To enforce the common assumption that the noise is isotropic, this estimate is then radially averaged, yielding a PSD function that depends only on the magnitude of the spatial frequency. 

To estimate the PSD of the observed data, we start by averaging the power spectra of all input images
\begin{equation}\label{eq:total-ps}
    \hat{P}_{YY}(\mathbf{f}) = \frac{1}{N} \sum_{i=1}^N |Y_i(\mathbf{f})|^2.
\end{equation}
Since the signal and noise are assumed to be uncorrelated, their PSDs add linearly such that
\[
P_{YY}(\mathbf{f}) = P_{SS}(\mathbf{f}) + P_{NN}(\mathbf{f}).
\]
We can therefore estimate the signal's PSD by
\begin{equation}
    \hat{P}_{SS}(\mathbf{f}) = \max(0, \hat{P}_{YY}(\mathbf{f}) - \hat{P}_{NN}(\mathbf{f})).
\end{equation}
The $\max(0, \cdot)$ operation ensures the non-negativity of the estimated power, correcting for statistical fluctuations where the estimated noise power might locally exceed the total power. Figure~\ref{subfig:projection-ps} shows the log of the power spectrum of the observed projections, as calculated by~\eqref{eq:total-ps}. Figures \ref{subfig:noise-ps} and \ref{subfig:denoised-ps} show the extracted noise and signal power spectra, respectively.

\begin{figure}
    \centering
    \begin{subfigure}[b]{0.25\textwidth}
        \centering
        \includegraphics[width=\textwidth]{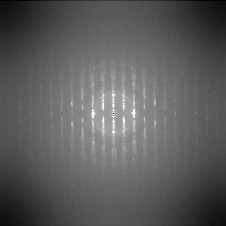}
        \caption{}
    \label{subfig:projection-ps}
    \end{subfigure}
    \hfill
    \begin{subfigure}[b]{0.25\textwidth}
        \centering
        \includegraphics[width=\textwidth]{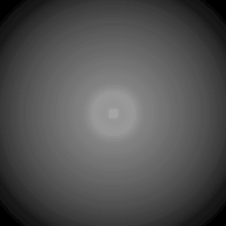}
        \caption{}
    \label{subfig:noise-ps}
    \end{subfigure}
    \hfill
    \begin{subfigure}[b]{0.25\textwidth}
        \centering
        \includegraphics[width=\textwidth]{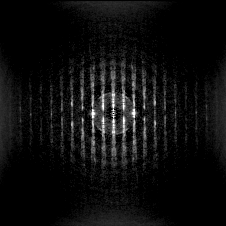}
        \caption{}
    \label{subfig:denoised-ps}
    \end{subfigure}
    \label{fig:power-spectra}
    \caption{Power spectra estimation for projections from the EMPIAR-10019 dataset. (a)~The (log) power spectrum of the projections. (b)~The estimated (log) power spectrum of the noise. (c)~The estimated (log) power spectrum of the signal.}
\end{figure}

With the estimated signal and noise PSDs, $\hat{P}_{SS}(\mathbf{f})$ and $\hat{P}_{NN}(\mathbf{f})$, we can now construct an empirical Wiener filter according to~\eqref{eq:wiener-filter-def}. This filter is then applied in the Fourier domain to each of the aligned segment images. The denoising procedure is demonstrated in Figure~\ref{fig:Wiener denoising}.

\begin{figure}[H]
    \centering
    \begin{subfigure}[b]{0.45\textwidth}
        \centering
        \includegraphics[width=0.5\textwidth]{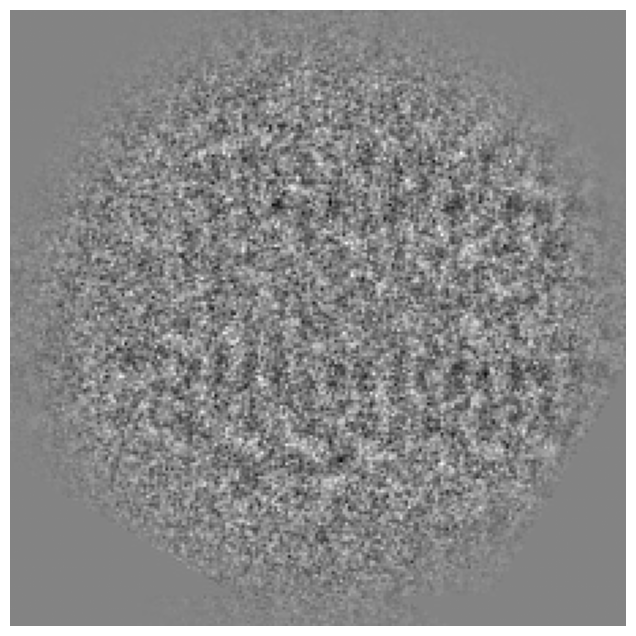}
        \caption{}
    \label{subfig:denoise-before}
    \end{subfigure}
    \hspace{0.02\textwidth} % Adjust this value as needed for spacing
    \begin{subfigure}[b]{0.45\textwidth}
        \centering
        \includegraphics[width=0.5\textwidth]{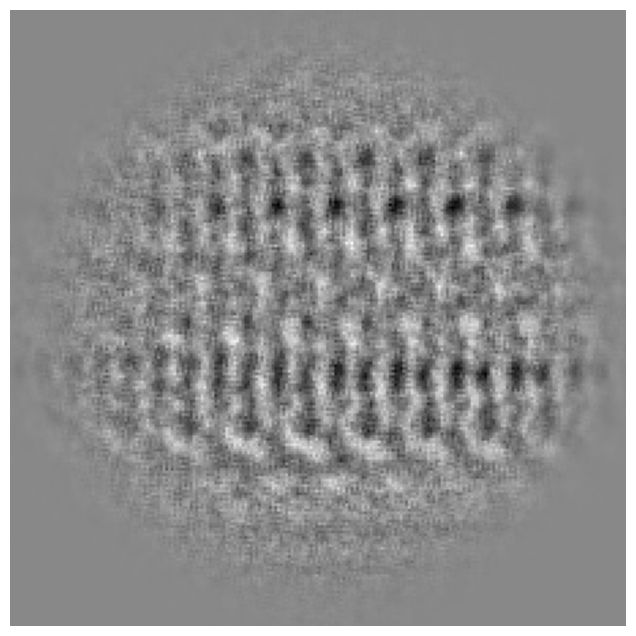}
        \caption{}
        \label{subfig:denoise-after}
    \end{subfigure}
    \caption{The effect of the denoising process on a projection from the EMPIAR-10019 dataset. (a)~The projection before denoising. (b)~The same projection after denoising.}
    \label{fig:Wiener denoising}
\end{figure}

\subsection{Spectral angle recovery}

This stage applies the SHREC algorithm (Algorithm~\ref{alg:angle-extraction-continuous}) to estimate the projection angles directly from the denoised image data.

To improve computational efficiency and reduce the impact of noise, we first reduce the dimensionality of the image data as follows:
\begin{enumerate}
    \item \textbf{Variance-based pixel selection}: We calculate the variance of each pixel across all projections and select those with variance above a certain threshold (typically the top 20-30\%). This focuses the analysis on regions of the image where meaningful structural variations occur.
    
    \item \textbf{PCA-based dimension reduction}: We compute the principal components of the variance-selected pixel data using a subset of the projections (for computational efficiency) and project all images onto the leading components (typically 256 components). This captures most of the variance in the data while significantly reducing dimensionality.
\end{enumerate}

After dimensionality reduction, we calculate the pairwise Euclidean distances between images
$R_{ij} = \|I_i - I_j\|_2$, 
where \(I_i\) and \(I_j\) are the dimensionality-reduced versions of the $i$-th and $j$-th images.
We then construct the weight matrix using a kernel function \(K: \mathbb{R} \rightarrow \mathbb{R}\), which is defined as
\begin{equation}\label{eq:kernel-definition}
K(R_{ij}) = 
\begin{cases}
\exp\left(-\dfrac{R_{ij}^2}{2\varepsilon}\right), & \text{if } I_j \in \mathcal{N}_k(I_i), \\
0, & \text{otherwise},
\end{cases}
\end{equation}
where \( \varepsilon > 0 \) is a bandwidth parameter, and \( \mathcal{N}_k(I_i) \) denotes the set of \( k \) nearest neighbors of image~\( I_i \). The bandwidth parameter~\(\varepsilon\) can be determined  by trial and error, but we use the following heuristic by default. We calculate the distance from each of the dimensionality-reduced images to its nearest neighbor, and set \(\varepsilon\) to be the 95th percentile of this set of distances.

Once we construct the kernel matrix using $k$ of~\eqref{eq:kernel-definition}, we use this matrix to construct the density-invariant graph Laplacian $\tilde{\mathbf{L}}$ and compute the eigenvectors corresponding to the smallest eigenvalues of $\tilde{\mathbf{L}}$ (see steps 3-4 in Algorithm~\ref{alg:angle-extraction-continuous}). Here, $k$ represents the number of nearest neighbors used in the kernel construction. However, for numerical reasons, instead of calculating them directly, we use the following observation. The matrix \(\tilde{\mathbf{D}}^{-1}\tilde{\mathbf{W}}\), described in Section~\ref{subsec:spectral-embedding}, is similar to the symmetric matrix \(\mathbf{S}=\tilde{\mathbf{D}}^{-\frac{1}{2}}\tilde{\mathbf{W}}\tilde{\mathbf{D}}^{-\frac{1}{2}}\), since
\begin{equation*}
    \tilde{\mathbf{D}}^{-1}\tilde{\mathbf{W}} = \tilde{\mathbf{D}}^{-\frac{1}{2}} \underbrace{\tilde{\mathbf{D}}^{-\frac{1}{2}}\tilde{\mathbf{W}}\tilde{\mathbf{D}}^{-\frac{1}{2}}}_{=\mathbf{S}} \tilde{\mathbf{D}}^{\frac{1}{2}} = \tilde{\mathbf{D}}^{-\frac{1}{2}}\mathbf{S}\tilde{\mathbf{D}}^{\frac{1}{2}}.
\end{equation*}
As such, $\tilde{\mathbf{D}}^{-1}\tilde{\mathbf{W}}$ has a complete set of eigenvectors \(\mathbf{v}_i\) with corresponding eigenvalues \(1 = \lambda_0 > \lambda_1 \ge ... \ge \lambda_{N-1}\). The graph Laplacian operator $\tilde{\mathbf{L}}$ has eigenvalues related to~\(\lambda_i\) by \(\mu_i = 1 - \lambda_i\), so that \(0 = \mu_0 < \mu_1 \le \dots \le \mu_{N-1}\). If \(\mathbf{u}_i\) is an eigenvector of \(\mathbf{S}\) with eigenvalue \(\lambda_i\), then \(\mathbf{v}_i = \tilde{\mathbf{D}}^{1/2}\mathbf{u}_i\) is the corresponding eigenvector of \(\tilde{\mathbf{L}}\) with eigenvalue \(\mu_i = 1 - \lambda_i\). Since computing the eigenvectors of a symmetric matrix is generally more stable, we compute the eigenvectors \(\mathbf{u}_1, \mathbf{u}_2\) corresponding to the highest eigenvalues of the matrix \(\mathbf{S}\), and then transform them to obtain the eigenvectors of \(\tilde{\mathbf{L}}\) by computing \(\mathbf{v}_i = \tilde{\mathbf{D}}^{1/2}\mathbf{u}_i\). As explained in Section \ref{subsec:spectral-embedding}, the first eigenvector corresponds to the constant function, and the next two eigenvectors (after the transformation \(\mathbf{v}_i = \tilde{\mathbf{D}}^{1/2}\mathbf{u}_i\)) are used as a spectral embedding for recovering the angles according to Algorithm \ref{alg:angle-extraction-continuous}. The technical details of computing the eigenvectors of $\tilde{\mathbf{L}}$ via matrix symmetrization are implemented as described in steps 3-4 of Algorithm~\ref{alg:angle-extraction-continuous}.

\subsection{3D reconstruction and refinement}

With the estimated projection angles $\{\theta_i\}$ in hand (see Algorithm \ref{alg:angle-extraction-continuous}), the pipeline transitions back to the RELION framework to generate a high-resolution 3D model. This process consists of the following steps.

\subsubsection{Initial model generation}\label{subsec:initial-model}
A critical step in our workflow is generating an initial 3D model using the SHREC-derived angles. To that end, we use a featureless cylindrical reference~\cite{RELION-helical-reconstruction}, and apply single-class 3D classification~\cite{RELION-paper}, with the SHREC-determined angles assigned to each helical segment as initial orientation parameters. Additionally, we set a highly restricted angular search range (typically within 6\degree\, of the SHREC angles). This approach leverages SHREC's angle assignments to create a consistent initial model, allowing for minor refinement of the angles within RELION’s Bayesian framework. 
    
\subsubsection{Helical parameters estimation} 
Once this initial 3D model is obtained, the helical symmetry parameters -- the rise~$\Delta x$ and twist~$\Delta \theta$ -- can be reliably estimated by direct measurement from the map or by using tools that utilize one of the following approaches:
\begin{itemize}
    \item \textbf{Exhaustive search} - This approach, implemented by \texttt{relion\_helix\_toolbox --search}, involves searching within a predefined range of values for the rise and twist. It is particularly effective when prior knowledge limits the search to a reasonably small parameters space.
    \item \textbf{Real-space lattice indexing} - This method directly analyzes the spatial arrangement of subunits in real space to determine the symmetry parameters. Sun et al.~\cite{HI3D} use this approach to implement an algorithm for extracting helical symmetry parameters from 3D volumes. 
\end{itemize}

\subsubsection{3D classification and refinement}
Using the initial 3D model generated in the previous subsection and the estimated helical parameters (rise $\Delta x$ and twist $\Delta \theta$), we apply RELION's 3D classification and 3D refinement procedures. Both procedures are used to enhance the resolution of an initial 3D model based on the projection images through iterative expectation-maximization \cite{RELION-paper, RELION-helical-reconstruction}. The distinction between classification and refinement is purpose-driven:
\begin{itemize}
    \item \textbf{3D classification} primarily addresses structural heterogeneity by dividing particles into multiple discrete classes. When structural variability is suspected, multi-class 3D classification can identify distinct conformational states. The classification can be performed with or without alignment, depending on whether orientation refinement is desired alongside classification.
    
    \item \textbf{3D refinement} focuses on achieving maximum resolution for a single structural state through iterative optimization of orientation and translation parameters.
\end{itemize}

For the SHREC pipeline, we use the estimated projection angles $\{\theta_i\}$ for a homogeneous subset of the dataset to generate an initial 3D model. This initial model, along with the estimated helical parameters ($\Delta x$, $\Delta \theta$), is then used as a starting point for high-resolution refinement with the full dataset. Both classification and refinement jobs can impose the helical structure and perform local searches around the estimated symmetry parameters. This iterative refinement process enhances the accuracy of both the orientation parameters and the symmetry parameters, ultimately improving the resolution of the final 3D reconstruction.

In cases where structural heterogeneity is present, the estimated parameters from SHREC can be used to first create homogeneous subsets through 3D classification, followed by a separate refinement of each subset.

%% Explanation about the RELION refine and 3D classification

\section{Results}\label{sec:results}

We assess our algorithm by applying the reconstruction pipeline on three datasets (EMPIAR-10019, EMPIAR-10022, EMPIAR-10869). The first two datasets were chosen since in He et al.~\cite{RELION-helical-reconstruction}, reconstructing a three-dimensional volume from them required prior knowledge of the rise and twist parameters. In contrast, throughout our experiments, the only initial information used is the specimen's axial symmetry group (i.e., \(C_n\)) and the outer radius of the molecule's tube, which is easily measurable.

To assess the quality of the angle recovery algorithm, we measure the quality of the resulting volume using the following metrics:
\begin{itemize}
    \item Obtained resolution, based on the Fourier Shell Correlation (FSC) curve of volumes reconstructed from two random subsets of projections, with a 0.143 FSC-cutoff~\cite{self-fsc, fsc-threshold-criteria}.
    \item Obtained resolution using the FSC curve of the reconstructed volume and a reference volume, based on 0.5 FSC-cutoff~\cite{fsc-threshold-criteria}. 
    \item A comparison of the estimated helical rise and twist parameters with the established values.
\end{itemize}

\subsection{EMPIAR-10022}

We applied our algorithm to the EMPIAR-10022 dataset \cite{EMPIAR-10021-paper}, which consists of 107 multi-frame micrographs (two additional micrographs were dropped due to technical issues) of the Tobacco Mosaic Virus, as well as box files with picked filaments' coordinates. The EMDB entry corresponding to this dataset is EMD-2842 (with a resolution of \(3.35 \si{\angstrom}\)).

The initial set of picked particles (extracted using the provided box files) consists of 19,054 helical segments, each of size 400×400 pixels. After performing 2D classification with six classes, we retained only the class with the highest estimated resolution (among three classes with equal resolution, we selected the one with the most particles; see Figure~\ref{fig:10022-classes}), which contains 3,023 segments. The particles of this class were used as input to the SHREC algorithm and for reconstructing the initial model. 

\begin{figure}
    \centering
    \includegraphics[width=\linewidth]{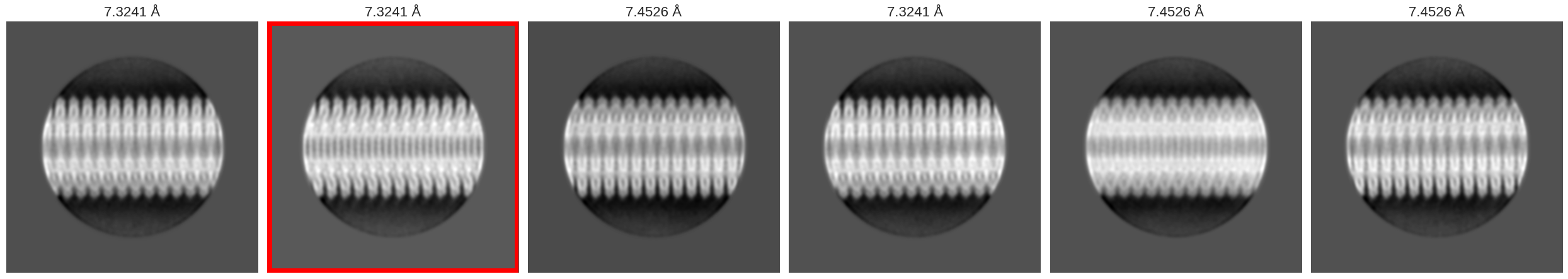}
    \caption{2D class averages with their estimated resolution. The average of the chosen class is marked with a red frame.}
    \label{fig:10022-classes}
\end{figure}

We applied the SHREC algorithm, using \(k = N/2\) nearest neighbors in the kernel function defined in~\eqref{eq:kernel-definition}, where \(N\) is the number of projections. The resulting embedding, shown in Figure~\ref{fig:10022-embedding}, exhibits a visible circular structure, as expected.

\begin{figure}
    \centering
    \includegraphics[width=0.4\textwidth]{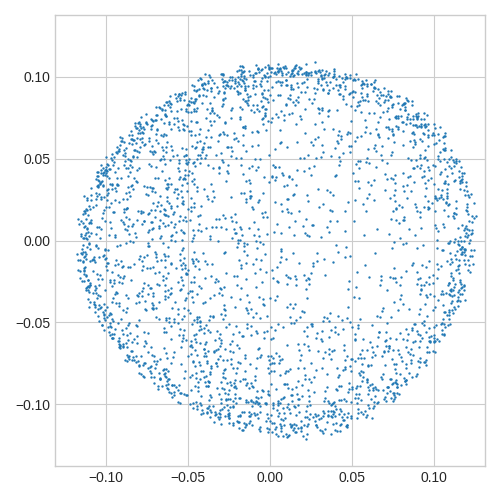}
    \caption{The 2D embedding of the selected 3,023 helical segments from EMPIAR-10022. The circular structure corresponds to the varying projection angles. These embedding coordinates are used to derive initial angle estimates for 3D reconstruction.}
    \label{fig:10022-embedding}
\end{figure}

\newpage

\begin{figure}
    \centering
    \includegraphics[width=0.38\textwidth]{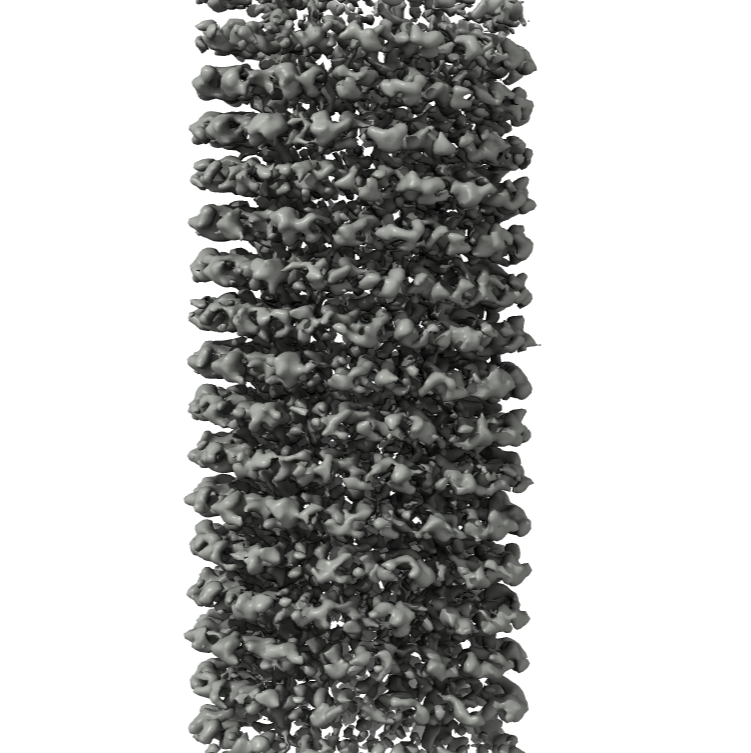}
    \caption{Initial volume reconstructed from the EMPIAR-10022 dataset.}
    \label{fig:10022-initial-volume}
\end{figure}

Using this initial embedding, we reconstructed an initial 3D volume (see Figure \ref{fig:10022-initial-volume}) from a featureless cylinder as described in Section~\ref{subsec:initial-model}. The angles derived from the embedding coordinates served as priors during a single-class 3D classification in RELION, starting from a cylindrical reference, as described in Section~\ref{subsec:initial-model}. This initial model was subsequently used as a reference for 3D refinement employing the entire dataset of 19,054 segments, but this time without using the angle priors (as they were only computed for the smaller subset). This strategy leverages the SHREC angles to generate a good starting model from a clean subset, then utilizes the full dataset for high-resolution refinement.

Based on this refined model, we estimated an approximate range for the helical symmetry parameters. The estimation was done by visual inspection of the volume and performing measurements using standard image processing software (ImageJ). Specifically, we counted the number of subunits per turn and measured the helix's pitch. The rise per subunit was then calculated by dividing the pitch by the number of subunits per turn, while the twist per subunit was determined by dividing 360° by the number of subunits per turn. Following these measurements, we performed 3D refinement incorporating a local symmetry search within RELION. The search range was set for the twist~\(\Delta \theta\) from \(-30^\circ\) to \(-19^\circ\) (reflecting the left-handedness observed in the initial reconstruction) and for the rise~\(\Delta x\) from \(1.27\,\si{\angstrom}\) to \(1.9\,\si{\angstrom}\). During the refinement process, the symmetry parameters converged to \(\Delta \theta = -22.035^\circ\) and \(\Delta x = 1.414\,\si{\angstrom}\). We then performed 3D classification with three classes and selected the class showing the highest resolution features for final refinement with the entire 19,054-image dataset. Finally, we performed masking and post-processing using RELION to obtain the final volume.

Based on the two independent half-maps from the 3D refinement process, the final volume after post-processing (see Figure~\ref{subfig:10022-volume-result}) has a resolution of 3.66Å (with 0.143 FSC cutoff). Since the volume we reconstructed has left-handed symmetry while the published volume (EMD-2842, see Figure \ref{subfig:10022-volume-published}) has right-handed symmetry, we flipped the former before comparing the two. We computed the FSC curve of our result against the published density map EMD-2842~\cite{EMPIAR-10021-paper}, and estimated the resolution with~0.5 FSC-cutoff. This yielded an estimated resolution of 3.9Å, compared to the published map with a resolution of 3.3Å (see Figure~\ref{fig:10022-fsc-comp})

During the final refinement, the helical parameters remained stable, converging to \(\Delta \theta = -22.036^\circ\) and \(\Delta x = 1.412\,\si{\angstrom}\), compared to \(\Delta \theta = 22.03^\circ\) and \(\Delta x = 1.408\,\si{\angstrom}\) of EMD-2842. Notice that both the rise and the magnitude of the twist we extracted match the published values. As explained in Section~\ref{sec:problem-setup}, reconstruction may result in a mirrored volume, and thus, we compare the magnitude of the twist.

\begin{figure}
    \centering
    \begin{subfigure}[b]{0.4\textwidth}
        \centering
        \includegraphics[width=\textwidth]{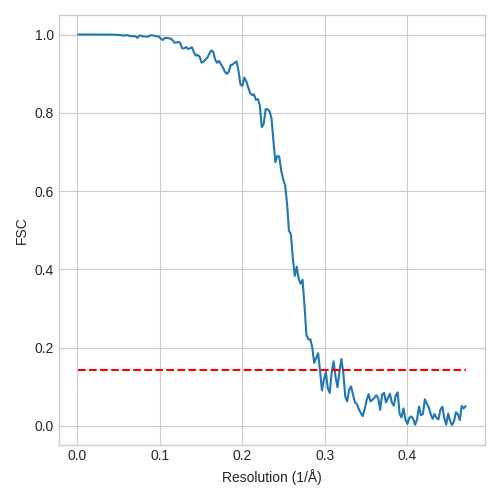}
    \caption{}
    \label{10022-fsc-self}
    \end{subfigure}
    \hfill
    \begin{subfigure}[b]{0.4\textwidth}
        \centering
        \includegraphics[width=\textwidth]{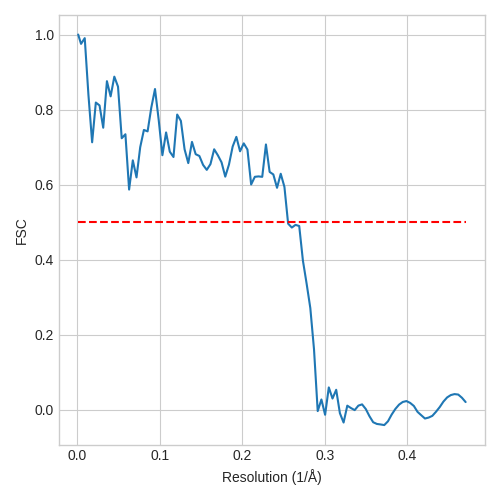}
        \caption{}
    \label{10022-compare-fsc-volume}
    \end{subfigure}
    \caption{(a) Half-maps FSC curve (\SI{3.66}{\angstrom}). (b) FSC curve comparing the reconstructed volume with EMD-2842 (\SI{3.9}{\angstrom}).}
    \label{fig:10022-fsc-comp}
\end{figure}

\begin{figure}
    \centering
    \begin{subfigure}[b]{0.4\textwidth}
        \centering
        \includegraphics[width=\textwidth]{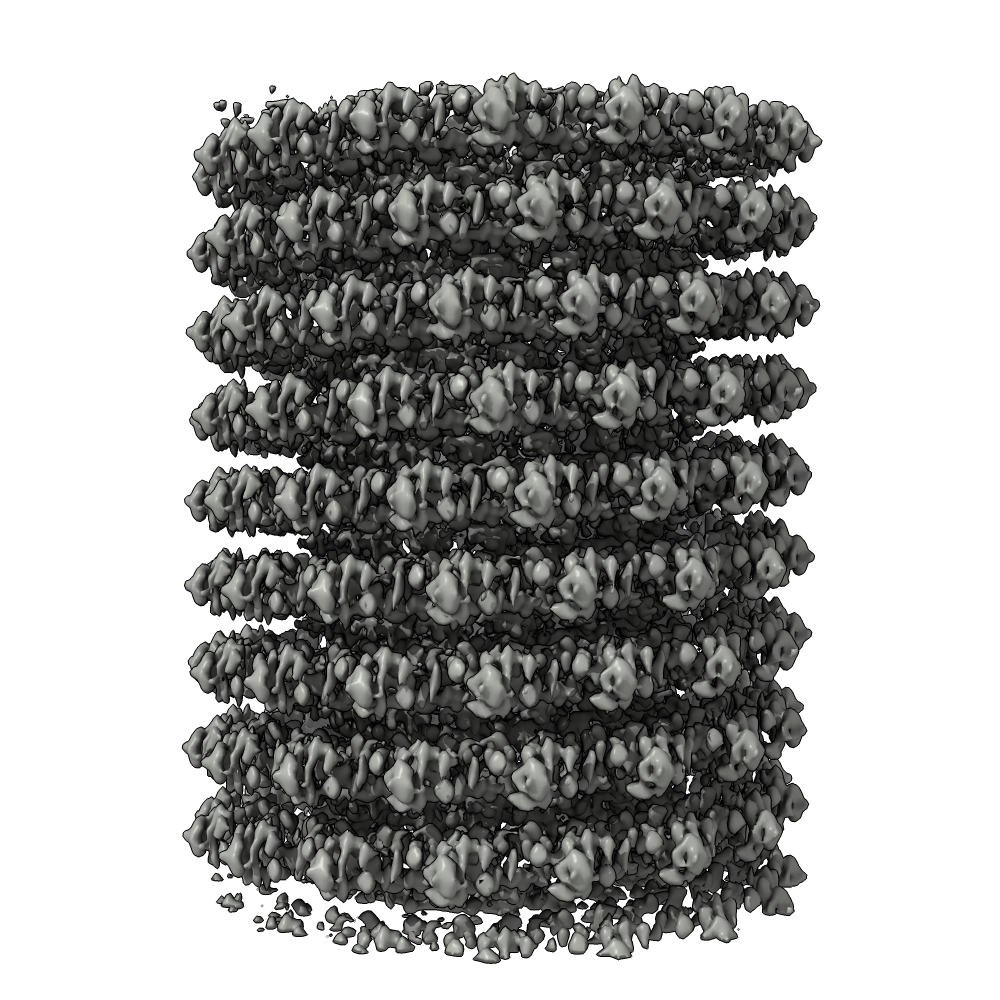}
    \caption{Reconstructed volume}
    \label{subfig:10022-volume-result}
    \end{subfigure}
    \hfill
    \begin{subfigure}[b]{0.4\textwidth}
        \centering
        \includegraphics[width=\textwidth]{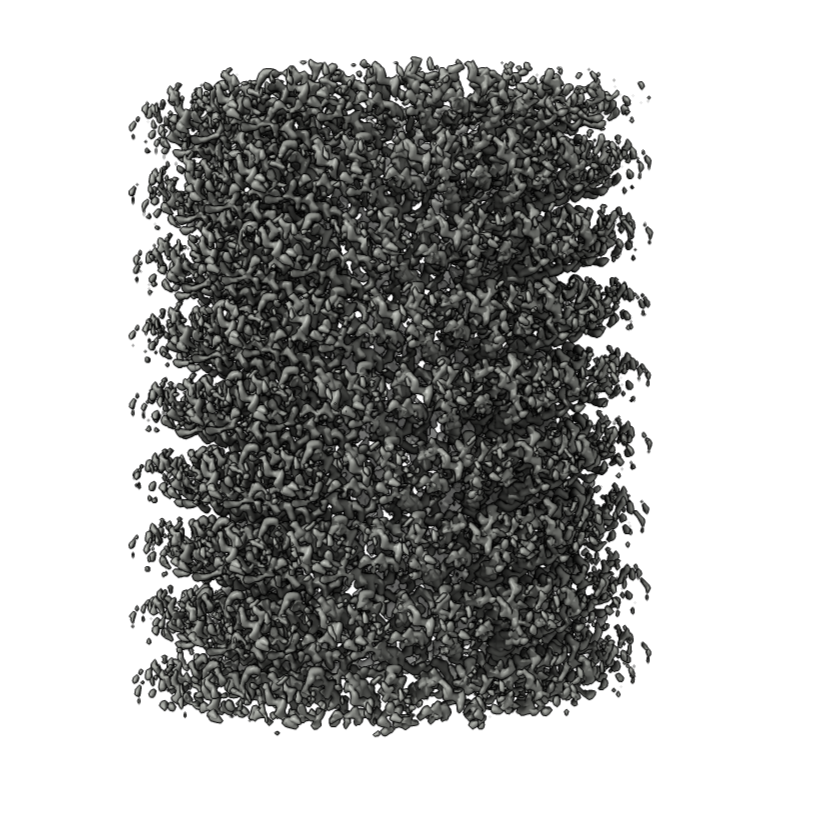}
        \caption{Published volume}
    \label{subfig:10022-volume-published}
    \end{subfigure}
    \label{fig:10022-volume}
    \caption{(a) The reconstructed volume from the EMPIAR-10022 dataset (b) The published volume EMD-2842.}
\end{figure}

\subsection{EMPIAR-10019}

We next applied our algorithm to the EMPIAR-10019 dataset~\cite{EMPIAR-10019-paper}, which consists of multi-frame micrographs of VipA/VipB sheath of the bacterial type IV secretion system. The EMDB entry corresponding to this dataset is EMD-2699 (with a resolution of \(3.5 \si{\angstrom}\)). The initial set of particles (extracted by manually picking the helical tube in RELION) comprised 15,896 helical segments. To facilitate the initial stages of the reconstruction process, we extracted projections while downsampling by a factor of 2 (2 Å/pixel instead of 1 Å/pixel), reducing the computational resources required. After performing 2D classification, we selected the class with the highest estimated resolution, containing 3,423 segments (see Figure~\ref{fig:enter-label}). These segments were then used as input for the SHREC algorithm and to generate the initial 3D model.

\begin{figure}
    \centering
    \includegraphics[width=\linewidth]{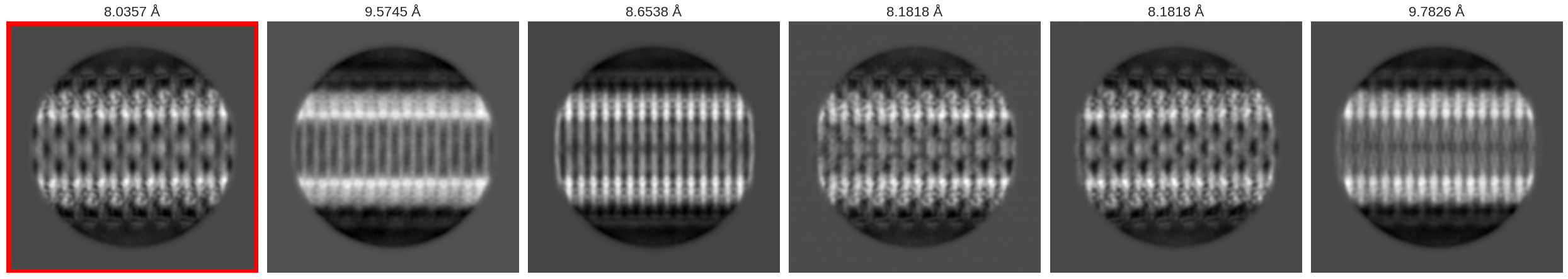}
    \caption{2D class averages with their estimated resolution. The average of the chosen class is marked with a red frame.}
    \label{fig:enter-label}
\end{figure}

\begin{figure}
    \centering
    \includegraphics[width=0.4\textwidth]{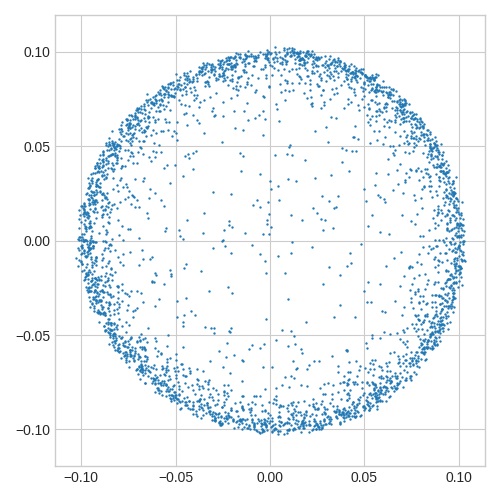}
    \caption{The 2D embedding of the selected 3,423 helical segments from EMPIAR-10019.}
    \label{fig:10019-embedding}
\end{figure}

We applied the SHREC algorithm and obtained an embedding, shown in Figure \ref{fig:10019-embedding}. After correcting for the \(C_6\) symmetry of the molecule, the resulting angles were used to reconstruct an initial 3D volume (see Figure \ref{fig:10019-initial-volume}). As with the EMPIAR-10022 dataset, the angles derived from the embedding coordinates served as priors. This initial 3D volume was then used as a reference for 3D refinement employing the entire dataset of 15,896 segments. Subsequently, we re-extracted the dataset at full resolution and performed another 3D refinement to obtain a final, high-resolution model. The initial volume we obtained was less clear visually than the one obtained from EMPIAR-10022, making it more difficult to measure the initial symmetry parameters directly. Therefore, we used the HI3D tool, which performs real-space helical indexing~\cite{HI3D}, and obtained estimates of the helical symmetry parameters: twist \(\Delta \theta\) of 30.49\degree, and rise \(\Delta x\) of 21.78\si{\angstrom}.

% \newpage
\begin{figure}
    \centering
\includegraphics[width=0.38\textwidth]{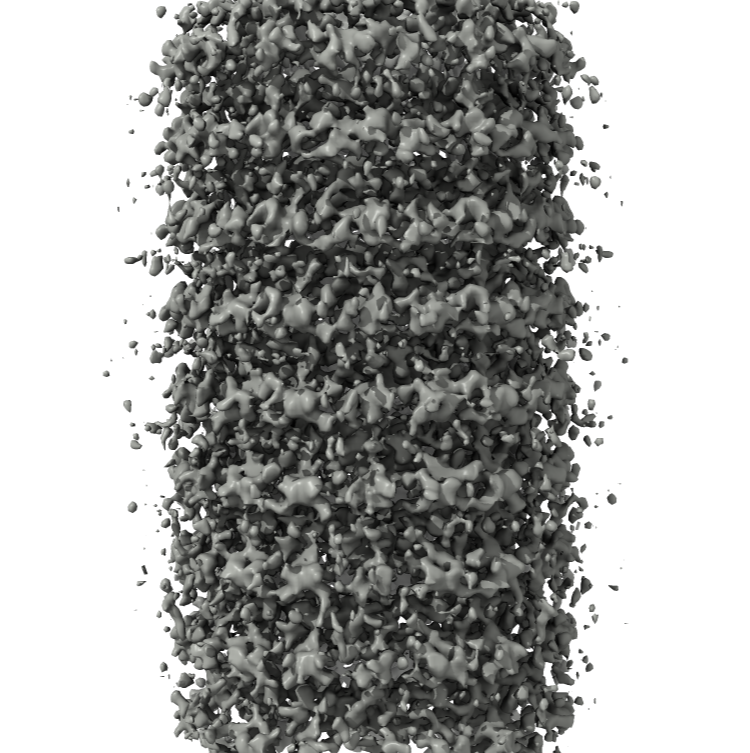}
    \caption{Initial volume reconstructed from the EMPIAR-10019 dataset.}
    \label{fig:10019-initial-volume}
\end{figure}

As in our analysis of the EMPIAR-10022 dataset, we performed 3D refinement incorporating a local symmetry search within RELION. The search range for the twist parameter \(\Delta \theta\) was 28.49\degree\ to 32.49\degree, and for the rise parameter \(\Delta x\), it was 20.78\si{\angstrom} to 22.78\si{\angstrom}. During refinement, the symmetry parameters converged to \(\Delta \theta\) = 29.41\degree\ and \(\Delta x\) = 21.78\si{\angstrom}, compared to the published \(\Delta \theta = 29.4^\circ\) \ and \(\Delta x = 21.78 \si{\angstrom}\). We performed masking and post-processing using RELION to obtain the final volume, which is shown in Figure~\ref{fig:10019-volume}.

Based on the two independent half-maps from the 3D refinement process, the final volume (see Figure \ref{subfig:10019-volume-result}) had a resolution of  3.66Å (with a 0.143 FSC cutoff; see Figure~\ref{subfig:10019-fsc-self}). We computed the FSC curve of our volume against the published density EMDB-2699 (see Figure~\ref{subfig:10019-volume-published}), and estimated the resolution with a 0.5 FSC cutoff. This yielded an estimated resolution of 4.0Å, compared to the published map's resolution of 3.5Å (see Figure~\ref{subfig:10019-compare-fsc-volume}).

\begin{figure}
    \centering
    \begin{subfigure}[b]{0.4\textwidth}
        \centering
        \includegraphics[width=\textwidth]{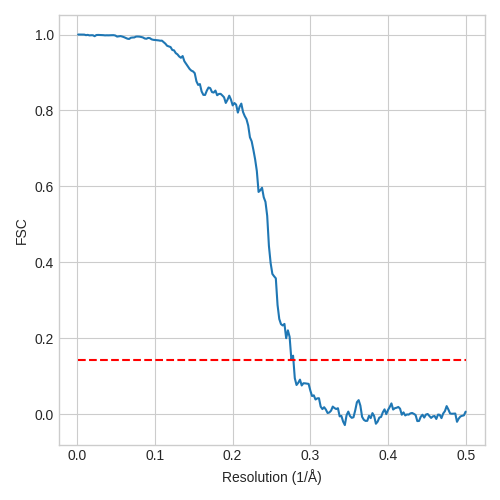}
    \caption{}
    \label{subfig:10019-fsc-self}
    \end{subfigure}
    \hfill
    \begin{subfigure}[b]{0.4\textwidth}
        \centering
        \includegraphics[width=\textwidth]{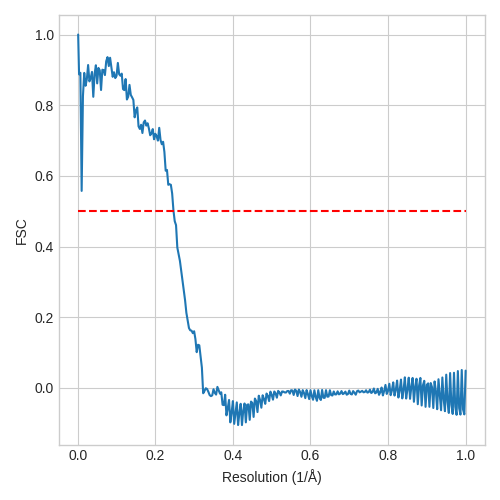}
        \caption{}
    \label{subfig:10019-compare-fsc-volume}
    \end{subfigure}
    \caption{(a) Half-maps FSC curve (3.66\AA) (b) FSC curve comparing the reconstructed volume with EMD-2699 (3.66\AA)}
    \label{fig:10019-fsc-comp}
\end{figure}

\begin{figure}
    \centering
    \begin{subfigure}[b]{0.4\textwidth}
        \centering
        \includegraphics[width=\textwidth]{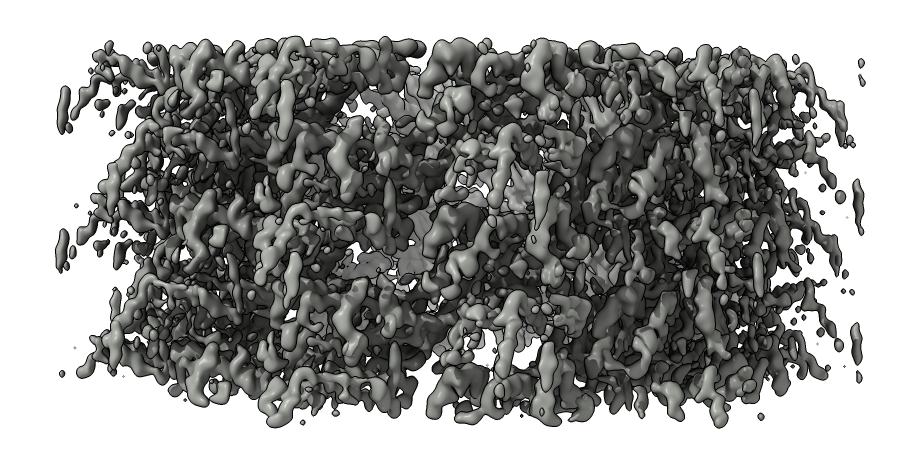}
    \caption{Reconstructed volume}
    \label{subfig:10019-volume-result}
    \end{subfigure}
    \hfill
    \begin{subfigure}[b]{0.4\textwidth}
        \centering
        \includegraphics[width=\textwidth]{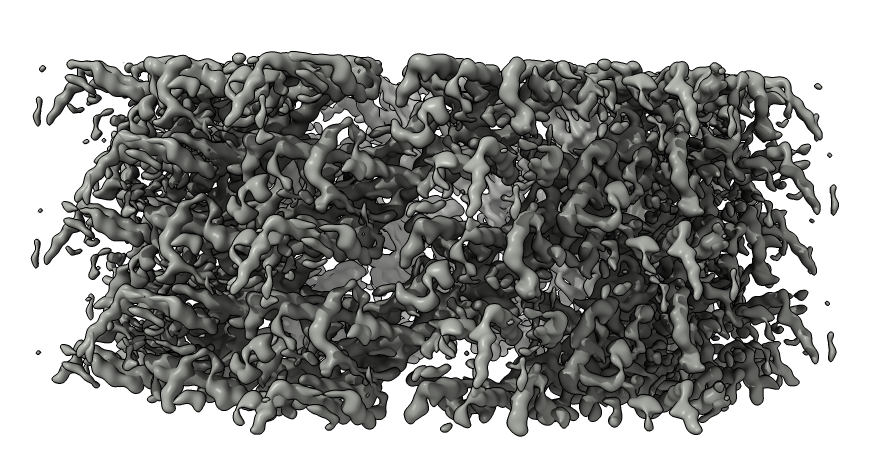}
        \caption{Published volume}
    \label{subfig:10019-volume-published}
    \end{subfigure}
    \caption{(a) The reconstructed volume from the EMPIAR-10019 dataset. (b) The published volume from EMDB (EMD-2699).}
    \label{fig:10019-volume}
\end{figure}

\subsection{EMPIAR-10869}

We next applied our algorithm to the EMPIAR-10869 dataset~\cite{EMPIAR-10869-paper}, which consists of 1,351 multi-frame micrographs of the helical toxin MakA from \textit{Vibrio cholerae}. The EMDB entry corresponding to this dataset is EMD-13185 (with a resolution of \(3.65\,\si{\angstrom}\)). The initial set of particles comprised 32,532 helical segments, extracted with a box size of 450 pixels and rescaled to 300 pixels for processing. After performing 2D classification with 32 classes, we selected two classes containing a total of 21,725 segments. These segments were used as input for the SHREC algorithm and to generate the initial 3D model.

We applied the SHREC algorithm, using \(k = N\) nearest neighbors in the kernel function defined in~\eqref{eq:kernel-definition}, where \(N\) is the number of projections. The resulting embedding, shown in Figure~\ref{fig:10869-embedding}, exhibits a visible circular structure, as expected.

\begin{figure}
    \centering
    \includegraphics[width=0.4\textwidth]{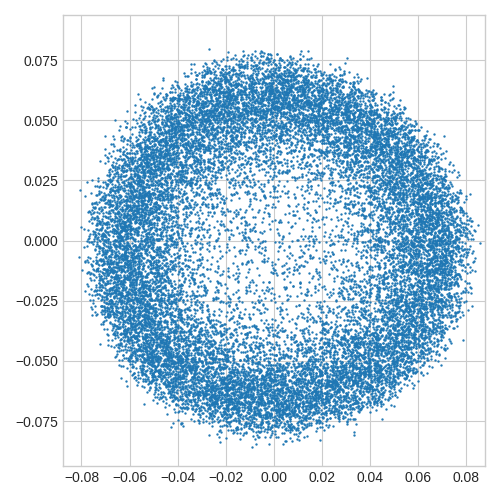}
    \caption{The 2D embedding of the selected 21,725 helical segments from EMPIAR-10869.}
    \label{fig:10869-embedding}
\end{figure}

Since the molecule possesses \(C_1\) symmetry, the embedding angles were used directly (without correction for cyclic symmetry) to reconstruct an initial 3D volume (see Figure~\ref{fig:10869-initial-volume}). As with the previous datasets, the angles derived from the embedding coordinates served as priors during a single-class 3D classification in RELION, starting from a cylindrical reference, as described in Section~\ref{subsec:initial-model}. This initial 3D volume was then used as a reference for 3D refinement employing the entire dataset of 32,532 segments. Subsequently, we re-extracted the dataset at full resolution and performed another 3D refinement to obtain a final, high-resolution model. We used the HI3D tool~\cite{HI3D} to estimate the helical symmetry parameters from the initial volume. Following these estimates, we performed 3D refinement incorporating a local symmetry search within RELION. The search range was set for the twist~\(\Delta\theta\) from \(-49.39^\circ\) to \(-47.39^\circ\) (reflecting the left-handedness observed in the initial reconstruction) and for the rise~\(\Delta x\) from \(5.72\,\si{\angstrom}\) to \(6.12\,\si{\angstrom}\). During refinement, the symmetry parameters converged to \(\Delta\theta = -48.594^\circ\) and \(\Delta x = 5.829\,\si{\angstrom}\). We performed masking and post-processing using RELION to obtain the final volume.

\begin{figure}
    \centering
    \includegraphics[width=0.38\textwidth]{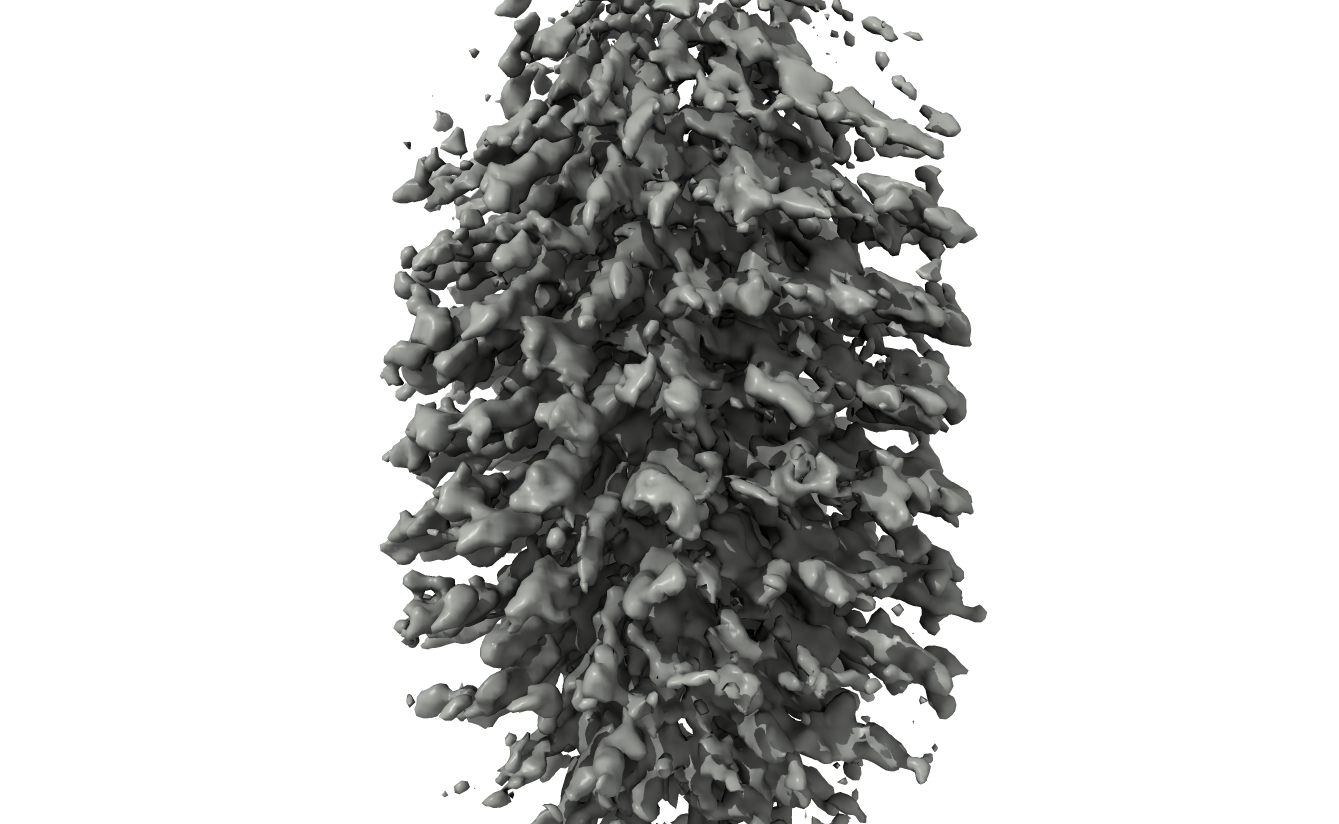}
    \caption{Initial volume reconstructed from the EMPIAR-10869 dataset.}
    \label{fig:10869-initial-volume}
\end{figure}

Based on the two independent half-maps from the 3D refinement process, the final volume after post-processing (see Figure~\ref{subfig:10869-volume-result}) has a resolution of 8.23\AA\ (with 0.143 FSC cutoff). Since the volume we reconstructed has left-handed symmetry while the published volume (EMD-13185, see Figure~\ref{subfig:10869-volume-published}) has right-handed symmetry, we flipped the former before comparing the two. We computed the FSC curve of our result against the published density map EMD-13185~\cite{EMPIAR-10869-paper}, and estimated the resolution with~0.5 FSC-cutoff. This yielded an estimated resolution of 8.0\AA, compared to the published map with a resolution of 3.65\AA\ (see Figure~\ref{fig:10869-fsc-comp}).

During the final refinement, the helical parameters converged to \(\Delta\theta = -48.594^\circ\) and \(\Delta x = 5.829\,\si{\angstrom}\), compared to \(\Delta\theta = 48.590^\circ\) and \(\Delta x = 5.841\,\si{\angstrom}\) of EMD-13185. As in the previous experiments, both the rise and the magnitude of the twist we extracted closely match the published values. As explained in Section~\ref{sec:problem-setup}, reconstruction may result in a mirrored volume, and thus, we compare the magnitude of the twist.

\begin{figure}
    \centering
    \begin{subfigure}[b]{0.4\textwidth}
        \centering
        \includegraphics[width=\textwidth]{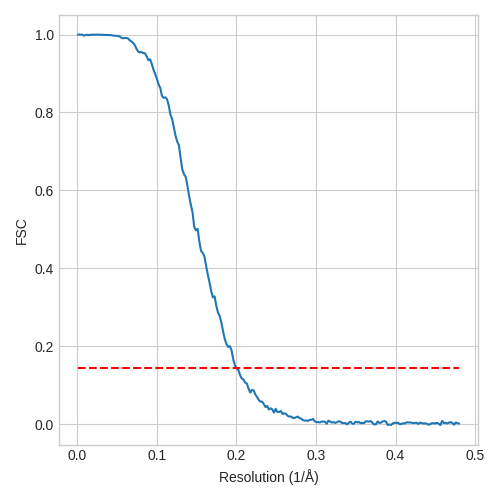}
    \caption{}
    \label{subfig:10869-fsc-self}
    \end{subfigure}
    \hfill
    \begin{subfigure}[b]{0.4\textwidth}
        \centering
        \includegraphics[width=\textwidth]{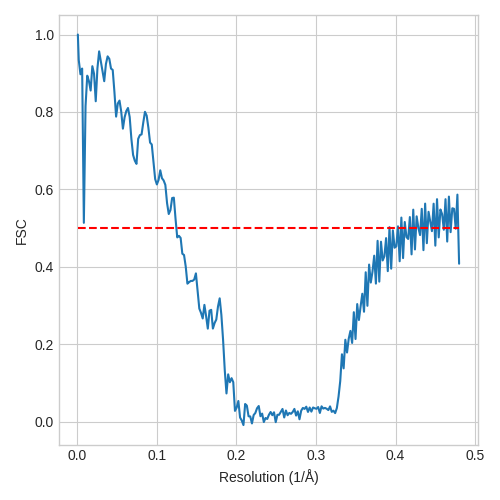}
        \caption{}
    \label{subfig:10869-compare-fsc-volume}
    \end{subfigure}
    \caption{(a) Half-maps FSC curve (\SI{8.23}{\angstrom}). (b) FSC curve comparing the reconstructed volume with EMD-13185 (\SI{8.0}{\angstrom}).}
    \label{fig:10869-fsc-comp}
\end{figure}

\begin{figure}
    \centering
    \begin{subfigure}[b]{0.4\textwidth}
        \centering
        \includegraphics[width=\textwidth]{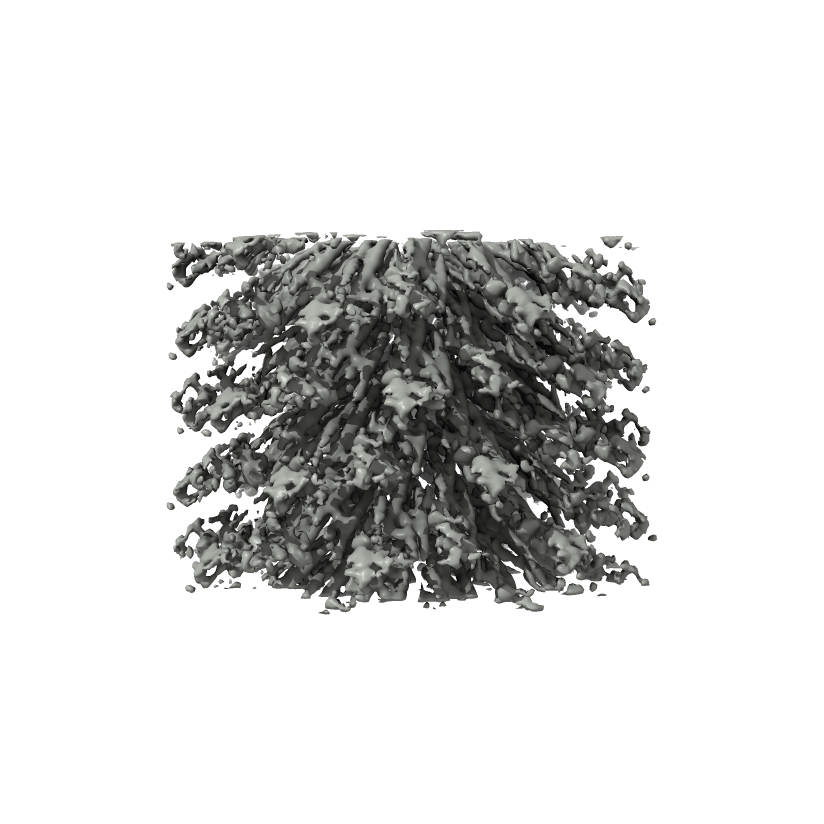}
    \caption{Reconstructed volume}
    \label{subfig:10869-volume-result}
    \end{subfigure}
    \hfill
    \begin{subfigure}[b]{0.4\textwidth}
        \centering
        \includegraphics[width=\textwidth]{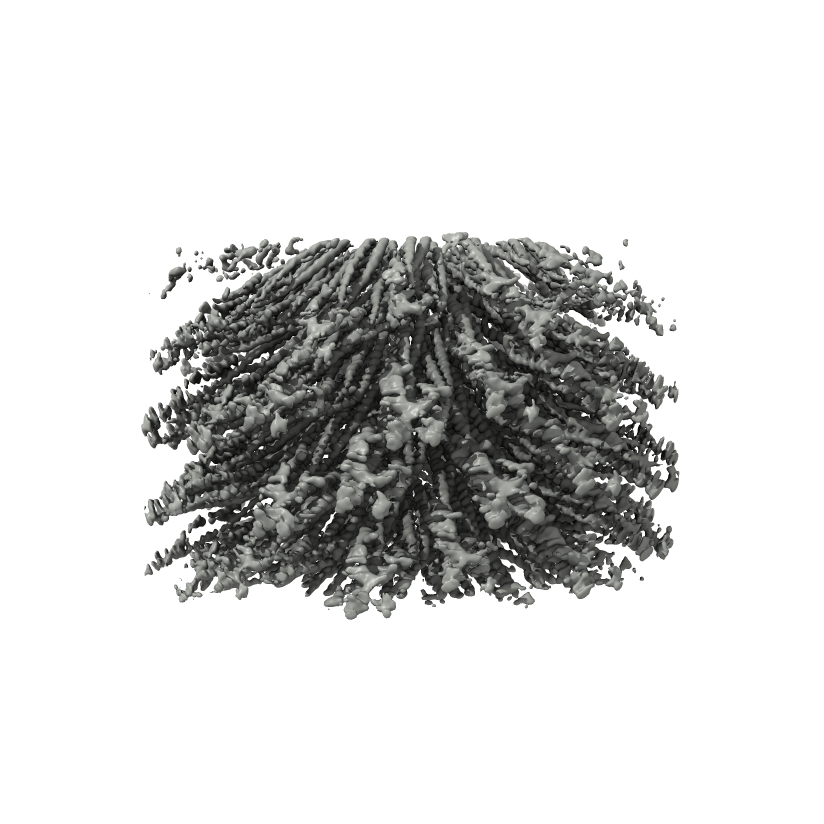}
        \caption{Published volume}
    \label{subfig:10869-volume-published}
    \end{subfigure}
    \caption{(a) The reconstructed volume from the EMPIAR-10869 dataset. (b) The published volume from EMDB (EMD-13185).}
    \label{fig:10869-volume}
\end{figure}

% \subsection{Summary of results}

% \begin{table}
%   \centering
%   \resizebox{\linewidth}{!}{%
%     \begin{tabular}{l c c}
%       \toprule
%       \textbf{Parameter} & \textbf{EMPIAR-10022 / EMD-2842} & \textbf{EMPIAR-10019 / EMD-2699} \\
%       \midrule
%       Dataset & EMPIAR-10022 & EMPIAR-10019 \\
%       Reference Map & EMD-2842 & EMD-2699 \\
%       \addlinespace
%       \multicolumn{3}{l}{\textbf{Helical Parameters (This Work / Published)}} \\
%       Rise (\si{\angstrom}) & \num{1.412} / \num{1.408} & \num{21.78} / \num{21.8} \\
%       Twist (\si{\degree}) & \num{-22.036} / \num{22.03}* & \num{29.41} / \num{29.4} \\ 
%       \addlinespace
%       \multicolumn{3}{l}{\textbf{Resolution (\si{\angstrom}) (This Work / Published)}} \\
%       FSC=0.143 (Half-Maps) & \num{3.66} / \num{3.35} & \num{3.66} / \num{3.5} \\
%       FSC=0.5 (vs. Reference) & \num{3.9} & \num{4.0}\\
%       \bottomrule
%       \multicolumn{3}{l}{\footnotesize *Note: Published twist has opposite sign due to handedness convention.} \\
%     \end{tabular}%
%   }
%   \caption{Summary of reconstruction results for the EMPIAR-10022 and EMPIAR-10019 datasets compared to the published structure EMD-2842 and EMD-2699 (respectively) \cite{EMPIAR-10021-paper, EMPIAR-10019-paper}. XXX There is no reference in the text to this table XXX}.
%   \label{tab:summary_results}
% \end{table}

\section{Conclusions and future work}
In this work, we presented SHREC (Spectral Helical Reconstruction), an algorithm for determining projection angles of helical segments in cryo-EM, without requiring prior knowledge of the helical symmetry parameters. By establishing that projections of helical segments form a one-dimensional manifold diffeomorphic to the circle, we demonstrated that spectral embedding techniques can reliably recover the underlying angular relationships between projections directly from the projection data. 

The experimental validation on publicly available datasets (EMPIAR-10019, EMPIAR-10022, and EMPIAR-10869) demonstrates that a reconstruction pipeline that utilizes the SHREC algorithm enables the successful reconstruction of helical structures that would otherwise require prior knowledge of symmetry parameters. Notably, our method requires only knowledge of the specimen's axial symmetry group ($C_n$) and approximate radius, avoiding the trial-and-error parameter searches or potentially biased initial estimates that characterize many existing workflows. The successful integration with the RELION software suite further demonstrates the practical applicability of our approach within established cryo-EM processing pipelines.

Several promising directions for future research emerge. First, while SHREC successfully generates initial 3D models that enable subsequent symmetry determination, the current pipeline still relies on external tools or manual measurements for extracting helical parameters from these models. Developing robust, automated methods for helical symmetry detection directly from the SHREC-generated initial volumes would create a more streamlined workflow. This could involve extending real-space indexing approaches or developing novel algorithms that leverage the known relationship between the spectral embedding structure and the underlying helical geometry. Such improvements would reduce user intervention and potential errors in symmetry determination, particularly for challenging cases with unusual helical parameters or cases where the helical symmetry is not clearly visible in the reconstruction due to low contrast or noise.

Additionally, the noise sensitivity of spectral methods remains an important consideration for very low SNR data. Future work could explore adaptive denoising strategies that preserve the manifold structure while suppressing noise, potentially leveraging recent advances in deep learning-based denoising. Additionally, investigating alternative similarity metrics beyond $L^2$ distance, such as cross-correlation or mutual information, may improve robustness for datasets with imaging variations.

\section*{Acknowledgments}
Molecular graphics and analyses performed with UCSF ChimeraX~\cite{https://doi.org/10.1002/pro.4792}, developed by the Resource for Biocomputing, Visualization, and Informatics at the University of California, San Francisco, with support from National Institutes of Health R01-GM129325 and the Office of Cyber Infrastructure and Computational Biology, National Institute of Allergy and Infectious Diseases.

\bibliographystyle{plain}
\bibliography{main.bib}

@article {EMPIAR-10021-paper,
	Title = {Seeing tobacco mosaic virus through direct electron detectors},
	Author = {Fromm, Simon A and Bharat, Tanmay A M and Jakobi, Arjen J and Hagen, Wim J H and Sachse, Carsten},
	DOI = {10.1016/j.jsb.2014.12.002},
	Number = {2},
	Volume = {189},
	Month = {2},
	Year = {2015},
	Journal = {Journal of structural biology},
	ISSN = {1047-8477},
	Pages = {87—97},
	URL = {https://europepmc.org/articles/PMC4416312},
}

@article{RELION-paper,
title = {A Bayesian View on Cryo-EM Structure Determination},
journal = {Journal of Molecular Biology},
volume = {415},
number = {2},
pages = {406-418},
year = {2012},
issn = {0022-2836},
doi = {https://doi.org/10.1016/j.jmb.2011.11.010},
url = {https://www.sciencedirect.com/science/article/pii/S0022283611012290},
author = {Sjors H.W. Scheres}
}

@ARTICLE{graph_laplacian_2d_tomography,

  author={Coifman, Ronald R. and Shkolnisky, Yoel and Sigworth, Fred J. and Singer, Amit},

  journal={IEEE Transactions on Image Processing}, 

  title={Graph Laplacian Tomography From Unknown Random Projections}, 

  year={2008},

  volume={17},

  number={10},

  pages={1891-1899},

  doi={10.1109/TIP.2008.2002305}}

@INPROCEEDINGS{random_proj,
  author={Wang, Lingda and Zhao, Zhizhen},
  booktitle={2019 IEEE International Conference on Image Processing (ICIP)}, 
  title={Two-Dimensional Tomography from Noisy Projection Tilt Series Taken at Unknown View Angles with Non-Uniform Distributi}, 
  year={2019},
  volume={},
  number={},
  pages={1242-1246},
  doi={10.1109/ICIP.2019.8803755}}

@article{RELION-helical-reconstruction,
title = {Helical reconstruction in RELION},
journal = {Journal of Structural Biology},
volume = {198},
number = {3},
pages = {163-176},
year = {2017},
issn = {1047-8477},
doi = {https://doi.org/10.1016/j.jsb.2017.02.003},
url = {https://www.sciencedirect.com/science/article/pii/S1047847717300199},
author = {Shaoda He and Sjors H.W. Scheres}
}

@phdthesis{RELION-helical-reconstruction-full, title={Helical Reconstruction in RELION}, url={https://www.repository.cam.ac.uk/handle/1810/284086}, DOI={10.17863/CAM.31456}, publisher={Apollo - University of Cambridge Repository}, author={He, Shaoda}, year={2018}, school={University of Cambridge}}

@Article{HI3D,
author={Sun, Chen
and Gonzalez, Brenda
and Jiang, Wen},
title={Helical Indexing in Real Space},
journal={Scientific Reports},
year={2022},
month={5},
day={17},
volume={12},
number={1},
pages={8162},
issn={2045-2322},
doi={10.1038/s41598-022-11382-7},
url={https://doi.org/10.1038/s41598-022-11382-7}
}

@article{IHRSR,
title = {The iterative helical real space reconstruction method: Surmounting the problems posed by real polymers},
journal = {Journal of Structural Biology},
volume = {157},
number = {1},
pages = {83-94},
year = {2007},
note = {Software tools for macromolecular microscopy},
issn = {1047-8477},
doi = {https://doi.org/10.1016/j.jsb.2006.05.015},
url = {https://www.sciencedirect.com/science/article/pii/S1047847706001687},
author = {Edward H. Egelman}
}

@Article{CryoSPARC-paper,
author={Punjani, Ali
and Rubinstein, John L.
and Fleet, David J.
and Brubaker, Marcus A.},
title={cryoSPARC: algorithms for rapid unsupervised cryo-EM structure determination},
journal={Nature Methods},
year={2017},
month={3},
day={01},
volume={14},
number={3},
pages={290-296},
issn={1548-7105},
doi={10.1038/nmeth.4169},
url={https://doi.org/10.1038/nmeth.4169}
}

@article{diffusion-maps-paper,
title = {Diffusion maps},
journal = {Applied and Computational Harmonic Analysis},
volume = {21},
number = {1},
pages = {5-30},
year = {2006},
note = {Special Issue: Diffusion Maps and Wavelets},
issn = {1063-5203},
doi = {https://doi.org/10.1016/j.acha.2006.04.006},
url = {https://www.sciencedirect.com/science/article/pii/S1063520306000546},
author = {Ronald R. Coifman and Stéphane Lafon}
}

@article{graph-laplacian-paper,
title = {From graph to manifold Laplacian: The convergence rate},
journal = {Applied and Computational Harmonic Analysis},
volume = {21},
number = {1},
pages = {128-134},
year = {2006},
note = {Special Issue: Diffusion Maps and Wavelets},
issn = {1063-5203},
doi = {https://doi.org/10.1016/j.acha.2006.03.004},
url = {https://www.sciencedirect.com/science/article/pii/S1063520306000510},
author = {A. Singer}
}

@Article{DEROSIER1968,
author={De Rosier, D. J.
and Klug, A.},
title={Reconstruction of Three Dimensional Structures from Electron Micrographs},
journal={Nature},
year={1968},
month={01},
day={01},
volume={217},
number={5124},
pages={130-134},
issn={1476-4687},
doi={10.1038/217130a0},
url={https://doi.org/10.1038/217130a0}
}

@ARTICLE{fourier-bessel-reconstruction,
  title    = "{Fourier-Bessel} reconstruction of helical assemblies",
  author   = "Diaz, Ruben and Rice, William J and Stokes, David L",
  journal  = "Methods Enzymol",
  volume   =  482,
  pages    = "131--165",
  year     =  2010,
  address  = "United States",
  language = "en"
}

@article{helical-reconstruction-again,
title = {Helical reconstruction, again},
journal = {Current Opinion in Structural Biology},
volume = {85},
pages = {102788},
year = {2024},
issn = {0959-440X},
doi = {https://doi.org/10.1016/j.sbi.2024.102788},
url = {https://www.sciencedirect.com/science/article/pii/S0959440X24000150},
author = {Edward H. Egelman}
}

@Article{symmetry-trap,
author ="Gambelli, Lavinia and Isupov, Michail N. and Daum, Bertram",
title  ="Escaping the symmetry trap in helical reconstruction",
journal  ="Faraday Discuss.",
year  ="2022",
volume  ="240",
issue  ="0",
pages  ="303-311",
publisher  ="The Royal Society of Chemistry",
doi  ="10.1039/D2FD00051B",
url  ="http://dx.doi.org/10.1039/D2FD00051B"}

@article{cheng2015primer,
  title={A primer to single-particle cryo-electron microscopy},
  author={Cheng, Yifan and Grigorieff, Nikolaus and Penczek, Pawel A and Walz, Thomas},
  journal={Cell},
  volume={161},
  number={3},
  pages={438--449},
  year={2015},
  publisher={Elsevier}
}

@article{lyumkis2019challenges,
  title={Challenges and opportunities in cryo-EM single-particle analysis},
  author={Lyumkis, Dmitry},
  journal={Journal of Biological Chemistry},
  volume={294},
  number={13},
  pages={5181--5197},
  year={2019},
  publisher={American Society for Biochemistry and Molecular Biology}
}

@article{singer2020computational,
  title={Computational methods for single-particle electron cryomicroscopy},
  author={Singer, Amit and Sigworth, Fred J},
  journal={Annual Review of Biomedical Data Science},
  volume={3},
  pages={163--190},
  year={2020},
  publisher={Annual Reviews}
}

@article{wu2020present,
  title={Present and emerging methodologies in cryo-EM single-particle analysis},
  author={Wu, Mengyu and Lander, Gabriel C},
  journal={Biophysical Journal},
  volume={119},
  number={7},
  pages={1281--1289},
  year={2020},
  publisher={Elsevier}
}

@article{fsc-threshold-criteria,
title = {Fourier shell correlation threshold criteria},
journal = {Journal of Structural Biology},
volume = {151},
number = {3},
pages = {250-262},
year = {2005},
issn = {1047-8477},
doi = {https://doi.org/10.1016/j.jsb.2005.05.009},
url = {https://www.sciencedirect.com/science/article/pii/S1047847705001292},
author = {Marin {van Heel} and Michael Schatz}
}

@Article{self-fsc,
author={Verbeke, Eric J.
and Gilles, Marc Aur{\`e}le
and Bendory, Tamir
and Singer, Amit},
title={Self Fourier shell correlation: properties and application to cryo-ET},
journal={Communications Biology},
year={2024},
month={1},
day={16},
volume={7},
number={1},
pages={101},
issn={2399-3642},
doi={10.1038/s42003-023-05724-y},
url={https://doi.org/10.1038/s42003-023-05724-y}
}

@Article{EMPIAR-10019-paper,
author={Kudryashev, Mikhail
and Wang, Ray Yu-Ruei
and Brackmann, Maximilian
and Scherer, Sebastian
and Maier, Timm
and Baker, David
and DiMaio, Frank
and Stahlberg, Henning
and Egelman, Edward H.
and Basler, Marek},
title={Structure of the Type VI Secretion System Contractile Sheath},
journal={Cell},
year={2015},
month={2},
day={26},
publisher={Elsevier},
volume={160},
number={5},
pages={952-962},
issn={0092-8674},
doi={10.1016/j.cell.2015.01.037},
url={https://doi.org/10.1016/j.cell.2015.01.037}
}

@article{cheng2022eigen,
  title={Eigen-convergence of Gaussian kernelized graph Laplacian by manifold heat interpolation},
  author={Cheng, Xiuyuan and Wu, Nan},
  journal={Applied and Computational Harmonic Analysis},
  volume={61},
  pages={132--190},
  year={2022},
  publisher={Elsevier}
}

@book{munkres2000topology,
  title     = {Topology},
  author    = {Munkres, James R.},
  edition   = {2nd},
  year      = {2000},
  publisher = {Prentice Hall},
  address   = {Upper Saddle River, NJ},
  isbn      = {978-0131816299},
  note      = {Theorem 26.6}
}

@article {EMPIAR-10869-paper,
article_type = {journal},
title = {Protein-lipid interaction at low pH induces oligomerization of the MakA cytotoxin from \textit{Vibrio cholerae}},
author = {Nadeem, Aftab and Berg, Alexandra and Pace, Hudson and Alam, Athar and Toh, Eric and Ådén, Jörgen and Zlatkov, Nikola and Myint, Si Lhyam and Persson, Karina and Gröbner, Gerhard and Sjöstedt, Anders and Bally, Marta and Barandun, Jonas and Uhlin, Bernt Eric and Wai, Sun Nyunt},
editor = {Egelman, Edward H and Aldrich, Richard W and Egelman, Edward H and Frost, Adam and Garcia-Saez, Ana J},
volume = 11,
year = 2022,
month = {2},
pub_date = {2022-02-08},
pages = {e73439},
citation = {eLife 2022;11:e73439},
doi = {10.7554/eLife.73439},
url = {https://doi.org/10.7554/eLife.73439},
journal = {eLife},
issn = {2050-084X},
publisher = {eLife Sciences Publications, Ltd},
}

@article{https://doi.org/10.1002/pro.4792,
	author = {Meng, Elaine C. and Goddard, Thomas D. and Pettersen, Eric F. and Couch, Greg S. and Pearson, Zach J. and Morris, John H. and Ferrin, Thomas E.},
	title = {UCSF ChimeraX: Tools for structure building and analysis},
	journal = {Protein Science},
	volume = {32},
	number = {11},
	pages = {e4792},
	doi = {https://doi.org/10.1002/pro.4792},
	url = {https://onlinelibrary.wiley.com/doi/abs/10.1002/pro.4792},
	eprint = {https://onlinelibrary.wiley.com/doi/pdf/10.1002/pro.4792},
	year = {2023}
}

\appendix
% \addcontentsline{toc}{chapter}{Appendices}

\begin{appendices}

\section{Additional proofs}\label{app:additional-proofs}

\begin{proof}[Proof of Lemma~\ref{lemma:reflection-invariance}]
We begin by expanding the right-hand side of~\eqref{eq:mirror-identity},
\[
P_{MRM} \psi^M(x, y) = \int_{-\infty}^{\infty} \psi^M\left(MRM(x, y, z)\right) \, \mathrm{d}z.
\]
By substituting~\eqref{eq:mirror-function}, this becomes
\[
P_{MRM} \psi^M(x, y) = \int_{-\infty}^{\infty} \psi\left(M \left(MRM(x, y, z)\right)\right) \, \mathrm{d}z.
\]
Next, we simplify the transformation inside the argument of \(\psi\). Observe that
\[
M M R M = I R M = R M,
\]
Thus,
\[
P_{MRM} \psi^M(x, y) = \int_{-\infty}^{\infty} \psi\left(R M(x, y, z)\right) \, \mathrm{d}z.
\]
Substituting the definition of the mirror map \(M(x, y, z) = (x, y, -z)\), we have
\[
P_{MRM} \psi^M(x, y) = \int_{-\infty}^{\infty} \psi\left(R(x, y, -z)\right) \, \mathrm{d}z.
\]
Now, substituting \(u = -z\), with \(\mathrm{d}u = -\mathrm{d}z\), results in
\[
\int_{-\infty}^{\infty} \psi\left(R(x, y, -z)\right) \, \mathrm{d}z = \int_{\infty}^{-\infty} \psi(R(x, y, u)) (-\mathrm{d}u) = \int_{-\infty}^{\infty} \psi(R(x, y, u)) \, \mathrm{d}u.
\]
This is exactly the definition of \(P_R\) from~\eqref{eq:ideal-projection-model}, hence,
\[
P_{MRM} \psi^M(x, y) = P_R \psi(x, y),
\]
as required.
\end{proof}

\vspace{1cm}

\begin{proof}[Proof of Lemma \ref{lemma:projection-differentiable}]

We start by showing that  $\Phi(\theta)\in L^2(\left[-\tfrac{B}{2},\tfrac{B}{2}\right]^2)$.
Notice that a cylinder centered around the \(x\)-axis is invariant under a rotation around the \(x\)-axis. This means that for any \(\theta \in S^1\)
\[
\mathrm{supp}(f\circ R_x(\theta)) \subseteq \left\{(x,y,z)\in \mathbb{R}^3 : \sqrt{y^2 + z^2} \le \frac{B}{2} \wedge |x| < \frac{B}{2} \right\}.
\]
% Since $f$ has compact support, there exists $r>0$ such that 
% \[
% \mathrm{supp}\,f\;\subseteq\;B_r(0)\subset\mathbb{R}^3.
% \]
Hence 
\[
\mathrm{supp}\,\Phi(\theta)\;\subseteq\;\left[-\tfrac{B}{2},\tfrac{B}{2}\right]^2
\]
and since $f \in C^1$ is continuously differentiable and has compact support, it is bounded. Therefore,
\[
\begin{aligned}
\|\Phi(\theta)\|_{L^{2}\left(\left[-\tfrac{B}{2}, \tfrac{B}{2}\right]^2\right)}^{2}
&= \int_{\left[-\tfrac{B}{2}, \tfrac{B}{2}\right]^2}
\left|\int_{\left[-\tfrac{B}{2}, \tfrac{B}{2}\right]} 
f\bigl(R_{x}(\theta)(x,y,z)^{T}\bigr) \, dz \right|^{2} dx\, dy \\
&\leq \int_{\left[-\tfrac{B}{2}, \tfrac{B}{2}\right]^3} 
\left| f\bigl(R_{x}(\theta)(x,y,z)^{T}\bigr) \right|^{2} dx\, dy\, dz \\
&\leq B^3 \|f\|_{L^{\infty}}.
\end{aligned}
\]
Thus, $\Phi(\theta)\in L^2(\left[-\tfrac{B}{2},\tfrac{B}{2}\right]^2)$.

\medskip

Next, we show that the derivative $\Phi'(\theta)$ exists. Since \(f \;\in\; C^1\bigl(\bigl[-\tfrac{B}{2},\tfrac{B}{2}\bigr]^3\bigr)\), we have that $\nabla f$ is continuous and therefore bounded. We set
\begin{equation}\label{g-definition}
g(u)\;=\;\|\nabla f\|_{L^\infty}\,\|u\|\;\chi_{\bigl(\bigl[-\tfrac{B}{2},\tfrac{B}{2}\bigr]^3\bigr)}(u),
\quad u=(x,y,z)^T.
\end{equation}
Since \(g\) is bounded on \(\bigl[-\tfrac{B}{2},\tfrac{B}{2}\bigr]^3\), we have that $g\in L^2(\left[-\tfrac{B}{2}, \tfrac{B}{2}\right]^3)$, and by the mean‐value theorem
\[
\frac{f(R_x(\theta+h)u)-f(R_x(\theta)u)}{h}
=\nabla f\bigl(R_x(\theta+\xi_h)u\bigr)\cdot\bigl(\dot R_x(\theta+\xi_h)u\bigr)
\quad(\xi_h\in(0,h)).
\]
Since \(\|\nabla f\bigl(R_x(\theta+\xi_h)u\bigr)\| \le \|\nabla f\|_{L^\infty}\) and \(\|\dot R_x(\theta+\xi_h)u\| \le \|\dot R_x(\theta+\xi_h)\|\,\|u\|\), and since \(\|\dot R_x(\vartheta)\| = 1\) (\(\|\cdot\|\) being the operator norm), for all \(\vartheta\) we have
\begin{equation}\label{eq:bounf by g}
\left |\frac{f(R_x(\theta+h)u)-f(R_x(\theta)u)}{h}\right |\;\le\;g(u).
\end{equation}
Now, we have that
\begin{align*}
\lim_{h\rightarrow 0}{\frac{\Phi(\theta+h)-\Phi(\theta)}{h}} &=\lim_{h\rightarrow 0}\frac{\int_{\mathbb{R}}f(R_x(\theta+h)u)-\int_{\mathbb{R}}f(R_x(\theta)u)}{h}dz\\
&=\lim_{h\rightarrow 0}\int_{\mathbb{R}}\frac{f(R_x(\theta+h)u)-f(R_x(\theta)u)}{h}dz.
\end{align*}
By the Dominated Convergence Theorem, since the integrand $\frac{f(R_x(\theta+h)u)-f(R_x(\theta)u)}{h}$ converges pointwise to $\nabla f(R_x(\theta)u) \cdot (\dot R_x(\theta)u)$ as $h \to 0$, and is uniformly bounded by $g(u)$ which is integrable with respect to $z$ for fixed $(x,y)$ (since $g \in L^2(\left[-\tfrac{B}{2}, \tfrac{B}{2}\right]^3)$ implies $\int_{-B/2}^{B/2} |g(x,y,z)|^2 dz < \infty$ for almost every $(x,y)$), we have that
\begin{equation*}
\lim_{h\rightarrow 0}{\frac{\Phi(\theta+h)-\Phi(\theta)}{h}}  = \int_{\mathbb{R}}\lim_{h\rightarrow 0}\frac{f(R_x(\theta+h)u)-f(R_x(\theta)u)}{h}dz.
\end{equation*}
with convergence in $L^2(\left[-\tfrac{B}{2}, \tfrac{B}{2}\right]^2)$. This means that the derivative $\Phi'(\theta)$ exists and that
\[
\Phi'(\theta)(x,y)
=\int_{-\infty}^{\infty}
\nabla f\bigl(R_x(\theta)(x,y,z)^T\bigr)\,\cdot\,\bigl(\dot R_x(\theta)(x,y,z)^T\bigr)\,dz.
\]

Finally, we show that $\Phi'(\theta)$ is continuous. Take $\theta_n\to\theta$ and define 
\[
H_{n}(u)=\bigl[\nabla f(R_{x}(\theta_{n})u)\cdot\bigl(\dot{R}_{x}(\theta_{n})u\bigr)-\nabla f(R_{x}(\theta)u)\cdot\bigl(\dot{R}_{x}(\theta)u\bigr)\bigr].
\]
From the continuity of $\nabla f$ on $\bigl[-\tfrac{B}{2},\tfrac{B}{2}\bigr]^3$ and the continuity of $R_x(\cdot)$ and $\dot R_x(\cdot)$ on $S^1$, we conclude that $H_n(\mathbf{u})\to0$ pointwise. Moreover, by the triangle inequality and the boundedness of $\nabla f$ and $\dot R_x$, we have
\begin{equation*}
\begin{aligned}
|H_n(u)|
&\le \big\|\nabla f(R_x(\theta_n)u)\big\|\,\big\|\dot R_x(\theta_n)u\big\|
     + \big\|\nabla f(R_x(\theta)u)\big\|\,\big\|\dot R_x(\theta)u\big\|\\
&\le \|\nabla f\|_{L^\infty}\,\|\dot R_x(\theta_n)\|\,\|u\|
   + \|\nabla f\|_{L^\infty}\,\|\dot R_x(\theta)\|\,\|u\|\\
&\le 2\,\|\nabla f\|_{L^\infty}\,\Big(\sup_{\vartheta\in S^1}\|\dot R_x(\vartheta)\|\Big)\,\|u\| \\
&= \|\nabla f\|_{L^\infty}\,\|u\|,
  \;=:\;2\,g(u),
\end{aligned}
\end{equation*}
where $g$ was defined in \eqref{g-definition}.

For any $(x,y) \in \left[-\tfrac{B}{2}, \tfrac{B}{2}\right]^2$, we have
\begin{align*}
\Phi'(\theta_n)(x,y) - \Phi'(\theta)(x,y) &= \int_{-B/2}^{B/2} H_n(x,y,z) \, dz.
\end{align*}
Since $H_n(u) \to 0$ pointwise with uniform bound $|H_n(u)| \le 2g(u)$ where $g \in L^2(\left[-\tfrac{B}{2}, \tfrac{B}{2}\right]^3)$, we can apply the Dominated Convergence Theorem as follows. For each fixed $(x,y) \in \left[-\tfrac{B}{2}, \tfrac{B}{2}\right]^2$, we have $H_n(x,y,z) \to 0$ pointwise in $z$, and $|H_n(x,y,z)| \le 2g(x,y,z)$ where $\int_{-B/2}^{B/2} |g(x,y,z)| \, dz < \infty$ for every $(x,y)$. Therefore, by the Dominated Convergence Theorem,
\begin{equation}\label{eq:Hn-convergence}
\int_{-B/2}^{B/2} H_n(x,y,z) \, dz \to 0 \quad \text{for a.e. } (x,y).
\end{equation}
Moreover, by Cauchy-Schwarz
\begin{equation}\label{eq:DCT-condition}
\left|\int_{-B/2}^{B/2} H_n(x,y,z) \, dz\right|^2 \le B^2 \int_{-B/2}^{B/2} |H_n(x,y,z)|^2 \, dz \le 4B^2 \int_{-B/2}^{B/2} |g(x,y,z)|^2 \, dz.
\end{equation}
Since the right-hand side of~\eqref{eq:DCT-condition} is integrable over $(x,y) \in \left[-\tfrac{B}{2}, \tfrac{B}{2}\right]^2$, we can apply the Dominated Convergence Theorem again, obtaining 
\begin{align*}
\lim_{n\to\infty} \|\Phi'(\theta_n)-\Phi'(\theta)\|_{L^2(\left[-\tfrac{B}{2}, \tfrac{B}{2}\right]^2)}^2 &= \lim_{n\to\infty} \int_{\left[-\tfrac{B}{2}, \tfrac{B}{2}\right]^2} \left|\int_{-B/2}^{B/2} H_n(x,y,z) \, dz\right|^2 dx\,dy \\
&=  \int_{\left[-\tfrac{B}{2}, \tfrac{B}{2}\right]^2} \lim_{n\to\infty} \left|\int_{-B/2}^{B/2} H_n(x,y,z) \, dz\right|^2 dx\,dy = 0,
\end{align*}
where the last equality follows since the limit of the integrand goes to zero due to~\eqref{eq:Hn-convergence}. Therefore, $\Phi'\colon S^1\to L^2(\left[-\tfrac{B}{2}, \tfrac{B}{2}\right]^2)$ is continuous. Therefore $\Phi\in C^1(S^1,L^2(\left[-\tfrac{B}{2}, \tfrac{B}{2}\right]^2))$, meaning $\Phi$ is a continuously differentiable map from the circle $S^1$ to the $L^2$ space of square-integrable functions on $\left[-\tfrac{B}{2}, \tfrac{B}{2}\right]^2$.
\end{proof}

\vspace{1cm}

\begin{proof}[Proof of Theorem \ref{theorem:manifold-structure}]
From Lemma \ref{lemma:projection-differentiable}, the map $\Phi:S^1\to L^2(\left[-\tfrac B2, \tfrac B2\right]^2)$ is continuously differentiable, with the derivative defined in \eqref{eq:Phi‐derivative}. Note that $\Phi(\theta) \in L^2(\left[-\tfrac{B}{2}, \tfrac{B}{2}\right]^2)$ follows from the continuity of $\Phi$ established in Lemma~\ref{lemma:projection-differentiable}. Since $f$ is defined on the compact domain $\left[-\tfrac{B}{2}, \tfrac{B}{2}\right]^3$, the condition $f \in C^1$ automatically implies $f \in L^2$ on this domain, so the condition $f \in C^1 \cap L^2$ is equivalent to $f \in C^1$. Assumption \textit{(ii)} of Theorem~\ref{theorem:manifold-structure} states that \(\Phi'(\theta) \neq 0\) for every \(\theta\), so the differential,
\[
d_\theta\Phi:T_\theta S^1 \rightarrow T_{\Phi(\theta)}L^2(\left[-\tfrac B2, \tfrac B2\right]^2),
\]
is injective at every point, since its kernel is trivial (as $\Phi'(\theta) \neq 0$ implies the derivative maps non-zero tangent vectors to non-zero vectors). Hence \(\Phi\) is an immersion. While the differential is everywhere injective, we still need to verify that \(\Phi\) itself is injective (up to the $C_n$ symmetry) to establish that it is an embedding.
To show that \(\Phi\) is locally an embedding (i.e., that every point has a neighborhood where \(\Phi\) restricted to that neighborhood is an embedding), we prove that \(\Phi\) is locally injective. Let \(\theta_0 \in S^1\), and let \(0 < h < \frac{2\pi}{n}\). From assumption \textit{(i)} of Theorem~\ref{theorem:manifold-structure}, the restriction of \(\Phi\) to the interval \(I_\theta = [\theta - h, \theta+h]\) is injective. Hence \(\Phi\) is locally injective and a local embedding.

Define the quotient map 
\[
q: S^1 \rightarrow S^1/C_n
\]
such that
\[
q(\theta) = [\theta] = \Bigl\{\theta + \frac{2\pi}{n}k : k=0,\dots,n-1\Bigr\}.
\]
Assumption \textit{(i)} states that two angles have the same projection if and only if they differ by a multiple of \(\frac{2\pi}{n}\), which is precisely the periodicity imposed by the \(C_n\) symmetry. Note that the quotient $S^1/C_n$ inherits a smooth $1$--manifold structure and is diffeomorphic to $S^1$ via the covering map $S^1 \to S^1/C_n$. 

Define \(\tilde{q}: S^1/C_n \rightarrow S^1\) such that
\[
    \tilde{q}([\theta]) = n\theta \,(\mathrm{mod}\,2\pi).
\]
The function \(\tilde{q}\) is well-defined because if \([\theta_1] = [\theta_2]\) in \(S^1/C_n\), then \(\theta_2 = \theta_1 + \frac{2\pi}{n}k\) for some \(k \in \{0,\ldots,n-1\}\), and thus \(n\theta_2 = n\theta_1 + 2\pi k \equiv n\theta_1 \pmod{2\pi}\). We define
\[
\tilde{\Phi}:S^1/C_n \rightarrow L^2(\left[-\tfrac B2, \tfrac B2\right]^2),\quad \tilde{\Phi}([\theta]) = \Phi(\theta) \text{ for } \theta \in [0,2\pi).
\]
The map $\tilde{\Phi}$ is well-defined because $\Phi$ factors through the quotient map $q$, i.e., $\Phi = \tilde{\Phi} \circ q$. Since $\Phi$ is~$C^1$ (from Lemma~\ref{lemma:projection-differentiable}) and has period $2\pi/n$, the induced map $\tilde{\Phi}: S^1/C_n \to L^2$ inherits smoothness from $\Phi$. The map $\tilde{q}$ provides an explicit diffeomorphism between $S^1/C_n$ and $S^1$, confirming that $S^1/C_n$ has the structure of a circle.
By the $C_n$ symmetry of $f$,
\[
\Phi\bigl(\theta + 2\pi/n\bigr)
= P_{R_x(\theta+2\pi/n)}f
= P_{R_x(\theta)}f
= \Phi(\theta).
\]
Thus $\Phi$ is $\tfrac{2\pi}{n}$-periodic.  On the other hand, assumption \textit{(i)} states there is no $\varphi\in[0,2\pi/n)$
for which $\Phi(\theta+\varphi)=\Phi(\theta)$. Hence, the minimal positive period of $\Phi$ is exactly $2\pi/n$ and~\(\tilde{\Phi}\) is injective.

Since $S^1 / C_n$ is compact and $L^2([-B/2,B/2]^2)$ is Hausdorff, a continuous injective map from $S^1 / C_n$ into $L^2$ is a homeomorphism onto its image (this follows from the general theorem that continuous bijections from compact spaces to Hausdorff spaces are homeomorphisms \cite{munkres2000topology}). Combined with the differentiability established in Lemma~\ref{lemma:projection-differentiable} and the non-vanishing derivative condition, this gives a smooth embedding.

Finally, since the domain $S^1/C_n$ is compact, the image $\tilde{\Phi}(S^1/C_n)$ is compact in the Hausdorff space $L^2([-B/2,B/2]^2)$, hence closed. Therefore, the image $\tilde{\Phi}(S^1/C_n)$ is a closed, one-dimensional~$C^1$ submanifold diffeomorphic to $S^1$.
\end{proof}

\vspace{1cm}

\begin{proof}[Proof of Lemma \ref{lemma:discrete-symmetry-equivalence}]
Set \(k = \lfloor\frac{t}{\Delta x} + \frac{1}{2}\rfloor\), where \(\lfloor \cdot \rfloor\) denotes the floor function. Define \(\theta_t = -k\Delta\theta\) and define \(s_t = t - k\Delta x\). By definition, \(s_t \in \left[-\frac{\rise}{2}, \frac{\rise}{2}\right)\). From the discrete helix property in Definition~\ref{def:helix} (part 2), 
\[
\psi(\mathbf{r}) = \psi(R_x(-k\Delta\theta)\mathbf{r} - k\Delta x\,\hat{\mathbf{x}}).
\]
Shifting both sides by \(t \hat{\mathbf{x}}\) yields
\[
\psi(\mathbf{r} + t \hat{\mathbf{x}}) = \psi(R_x(-k\Delta\theta)\mathbf{r} + t \hat{\mathbf{x}} - k\Delta x\,\hat{\mathbf{x}}) = \psi(R_x(\theta_t)\mathbf{r} + s_t\,\hat{\mathbf{x}}),
\]
which completes the proof.
\end{proof}

\vspace{1cm}

\begin{proof}[Proof of Theorem \ref{theorem:manifold-distance-bound}]
By Lemma~\ref{lemma:discrete-symmetry-equivalence}, for any \(t \in \mathbb R\) there exist \(\theta_t \in S^1\) and \(s_t \in \bigl[-\tfrac{\Delta x}{2},\tfrac{\Delta x}{2}\bigr)\) such that for all $\mathbf r$,
\[
\psi(\mathbf r + t\hat{\mathbf{x}})
\;=\;
\psi\bigl(R_x(\theta_t)\,\mathbf r + s_t\,\hat{\mathbf{x}}\bigr).
\]
Hence, the projection at axial shift $t$ may be written as
\[
\Pi(t)
= P_I\,S_B(t,\psi)
= P_I\bigl(S_B(0,\psi(\cdot + t\hat x))\bigr)
= P_I\bigl(S_B(0,\psi_{s_t}(R_x(\theta_t)\,\cdot))\bigr)
= P_{R_x(\theta_t)}\,S_B\bigl(0,\psi_{s_t}\bigr),
\]
where we set $\psi_{s_t}(\mathbf r)=\psi(\mathbf r + s_t\hat x)$.  

By definition,
\begin{equation*}
\begin{aligned}
d_{L^2(\Omega_B)}\bigl(\Pi(t),\mathcal M\bigr)
&=\inf_{\theta\in[0,2\pi)}\,
\bigl\|\,\Pi(t)-P_{R_x(\theta)}S_B(0,\psi)\bigr\|_{L^2}
\; \\ & \le\;
\bigl\|\,
P_{R_x(\theta_t)}S_B(0,\psi_{s_t})
\;-\;
P_{R_x(\theta_t)}S_B(0,\psi)
\bigr\|_{L^2}.
\end{aligned}
\end{equation*}
By using the integral triangle inequality, we obtain that for all \(\mathbf{r} \in \mathbb{R}^3\)
\begin{equation*}
    \begin{aligned}
    \left|\psi_{s_{t}}(\mathbf{r}) - \psi(\mathbf{r})\right|
    &= \left|\psi(\mathbf{r} + s_t \hat{\mathbf{x}}) - \psi(\mathbf{r})\right| \\
    &= \left|\int_{0}^{s_{t}} \partial_{x} \psi(\mathbf{r}) \, dx \right| \\
    &\leq \int_{0}^{s_{t}} \left| \partial_{x} \psi(\mathbf{r}) \right| \, dx \\
    &\leq |s_{t}| M_{x}(\psi),
\end{aligned}
\end{equation*}
and since \(|s_t| \le \frac{\rise}{2}\), we get that
\begin{equation*}
\max_{\mathbf{r} \in \mathbb{R}^3}{|\psi_{s_t}(\mathbf{r}) - \psi(\mathbf{r})|} \le \frac{1}{2}\rise M_{x}(\psi).
\end{equation*}
Now, for any $\mathbf{r} \in Q_B = [-\frac{B}{2}, \frac{B}{2}]^3$, we have
\begin{equation*}
|S_B(0,\psi_{s_t})(\mathbf{r}) - S_B(0,\psi)(\mathbf{r})| = |\psi_{s_t}(\mathbf{r}) - \psi(\mathbf{r})| \le \frac{1}{2}\rise M_{x}(\psi),
\end{equation*}
where we used the definition of the segment operator from Definition~\ref{def:helical-segment}. Therefore,
\begin{equation*}
\begin{aligned}
\bigl\|\,
P_{R_x(\theta_t)}S_B(0,\psi_{s_t})
\;-\;
P_{R_x(\theta_t)}S_B(0,\psi)
\bigr\|_{L^2(\Omega_B)}^2 &= \\\int_{\Omega_B} \left|\int_{-\frac{B}{2}}^{\frac{B}{2}}{\left(S_B(0,\psi_{s_t})(R_x(\theta_t)(x,y,z)^T)-S_B(0,\psi)(R_x(\theta_t)(x,y,z)^T)\right)dz}\right|^2 dxdy
&\le \\\int_{Q_B}\left|S_B(0,\psi_{s_t})(\mathbf{r})-S_B(0,\psi)(\mathbf{r})\right|^2d\mathbf{r} 
&\le \\B^3\left(\frac{1}{2}\rise M_{x}(\psi)\right)^2.
\end{aligned}
\end{equation*}
The integrals are known to be finite by \eqref{eq:integral-condition}. We get
\begin{equation*}
\bigl\|\,
P_{R_x(\theta_t)}S_B(0,\psi_{s_t})
\;-\;
P_{R_x(\theta_t)}S_B(0,\psi)
\bigr\|_{L^2} \le \frac{1}{2} \rise M_{x}(\psi)B^{\frac{3}{2}},
\end{equation*}
which completes the proof.
\end{proof}

\section{The Wiener filter}\label{app:wiener-filter}

We consider the problem of estimating a signal corrupted by additive noise in the frequency domain. Let \( S(\mathbf{f}) \) denote the true unknown signal and let \( N(\mathbf{f}) \) denote an additive noise component, both defined in the frequency domain \( \mathbf{f} \in \mathbb{R}^d \), and are part of a suitable function space (complex square-integrable functions). The observed data $Y(\mathbf{f})$ in the frequency domain is given by
\[
Y(\mathbf{f}) = S(\mathbf{f}) + N(\mathbf{f}).
\]
We treat \( S(\mathbf{f}) \) and \( N(\mathbf{f}) \) as zero-mean complex-valued random signals, and assume that \( N(\mathbf{f}) \) is uncorrelated with \( S(\mathbf{f}) \). The objective is to construct an estimator \( \widetilde{S}(\mathbf{f}) \) of \( S(\mathbf{f}) \) based only on the observation \( Y(\mathbf{f}) \), such that the mean squared error (MSE)
\[
\mathbb{E}\left[ \left| \widetilde{S}(\mathbf{f}) - S(\mathbf{f}) \right|^2 \right]
\]
is minimized. We further assume that the power spectral densities of \( S(\mathbf{f}) \) and \( N(\mathbf{f}) \), which are define as \( P_{SS}(\mathbf{f}) = \mathbb{E}[S(\mathbf{f}) S^*(\mathbf{f})] \) and \( P_{NN}(\mathbf{f}) = \mathbb{E}[N(\mathbf{f}) N^*(\mathbf{f})] \), respectively, are known.
Under these assumptions, we seek to find a filter $G(\mathbf{f})$ that minimizes 
\[
\mathbb{E}\left[ \left| G(\mathbf{f}) Y(\mathbf{f}) - S(\mathbf{f}) \right|^2 \right].
\]
This $G(\mathbf{f})$ is called the Wiener filter, and it produces the optimal linear estimator in the MSE sense.

To determine the Wiener filter $G(\mathbf{f})$, we apply the orthogonality principle, which states that the estimation error should be uncorrelated with the observed data, that is,
\[
\mathbb{E} \left[ \left(S(\mathbf{f}) - G(\mathbf{f}) Y(\mathbf{f})\right) Y^*(\mathbf{f}) \right] = 0,
\]
where \( Y^*(\mathbf{f}) \) denotes the complex conjugate of \( Y(\mathbf{f}) \). Expanding this expression gives
\[
\mathbb{E}[S(\mathbf{f}) Y^*(\mathbf{f})] = G(\mathbf{f}) \, \mathbb{E}[Y(\mathbf{f}) Y^*(\mathbf{f})].
\]
Using the model \( Y(\mathbf{f}) = S(\mathbf{f}) + N(\mathbf{f}) \) and the assumption that the signal and noise are uncorrelated, we get
\[
\mathbb{E}[S(\mathbf{f}) Y^*(\mathbf{f})] = \mathbb{E}[S(\mathbf{f}) S^*(\mathbf{f})] + \mathbb{E}[S(\mathbf{f}) N^*(\mathbf{f})] = P_{SS}(\mathbf{f}),
\]
and
\[
\mathbb{E}[Y(\mathbf{f}) Y^*(\mathbf{f})] = \mathbb{E}[(S(\mathbf{f}) + N(\mathbf{f}))(S^*(\mathbf{f}) + N^*(\mathbf{f}))] = P_{SS}(\mathbf{f}) + P_{NN}(\mathbf{f}),
\]
where \( P_{SS}(\mathbf{f}) \) and \( P_{NN}(\mathbf{f}) \) are the power spectral densities of the signal and the noise, respectively.
Solving for \( G(\mathbf{f}) \), we obtain 
\begin{equation*}
G(\mathbf{f}) = \frac{P_{SS}(\mathbf{f})}{P_{SS}(\mathbf{f}) + P_{NN}(\mathbf{f})}.
\end{equation*}
%The latter frequency-domain representation of the Wiener filter yields the optimal linear estimator of the signal, in the sense of minimizing mean squared error, when the observation is corrupted by additive uncorrelated noise.
\end{appendices}
\end{document}